\documentclass[runningheads,a4paper,10pt]{llncs}

\usepackage[
margin=2.84cm,
includefoot,
footskip=30pt,
]{geometry}

\usepackage{microtype}

\let\proof\relax
\let\endproof\relax

\usepackage{amsmath,amsthm,amssymb}

\newtheorem{fact}{Fact}

\usepackage{array}
\usepackage{multirow}

\usepackage{microtype}

\usepackage{amssymb}
\usepackage{graphicx}

\usepackage{fancybox}
\usepackage{minibox}
\usepackage[export]{adjustbox}
\usepackage[most]{tcolorbox}
\usepackage{ragged2e}
\usepackage{sgame}
\usepackage{graphicx}
\usepackage{tcolorbox}
\usepackage{amsmath}
\usepackage{amsfonts}
\usepackage{amsthm}
\usepackage{verbatim}
\usepackage{bbm}

\usepackage{mathtools}
\usepackage{cases}
\usepackage{color,soul}
\usepackage{esvect}
\usepackage[nospace,noadjust]{cite}
\usepackage{todonotes}
\usepackage{pgfplots}
\pgfplotsset{width=14cm,compat=1.9}
\usepackage{algorithm}
\usepackage[noend]{algpseudocode}
\usepackage{float}
\usepackage{wrapfig}
\usepackage{enumerate}

\usepackage{authblk}

\begin{document}

\title{Strategic Contention Resolution in Multiple Channels
}

\titlerunning{Strategic Contention Resolution in Multiple Channels}

\author{George Christodoulou\inst{1} \and Themistoklis Melissourgos\inst{1} \and Paul G. Spirakis\inst{1,2}\thanks{The work of this author was partially supported by the ERC Project ALGAME.}}

\authorrunning{G. Christodoulou, T. Melissourgos, and P. G. Spirakis}

\institute{Department of Computer Science, University of Liverpool, UK\\
	\email{\{G.Christodoulou,T.Melissourgos,P.Spirakis\}@liverpool.ac.uk}
	\and
	Computer Engineering \& Informatics Department, University of Patras, Greece
}

\maketitle

\begin{abstract}
	We consider the problem of resolving contention in communication networks with selfish users. In a \textit{contention game} each of $n \geq 2$ identical players has a single information packet that she wants to transmit using one of $k \geq 1$ multiple-access channels. To do that, a player chooses a slotted-time protocol that prescribes the probabilities with which at a given time-step she will attempt transmission at each channel. If more than one players try to transmit over the same channel (collision) then no transmission happens on that channel. Each player tries to minimize her own expected \textit{latency}, i.e. her expected time until successful transmission, by choosing her protocol. The natural problem that arises in such a setting is, given $n$ and $k$, to provide the players with a common, anonymous protocol (if it exists) such that no one would unilaterally deviate from it (equilibrium protocol). 
	
	All previous theoretical results on strategic contention resolution examine only the case of a single channel and show that the equilibrium protocols depend on the feedback that the communication system gives to the players. Here we present multi-channel equilibrium protocols in two main feedback classes, namely \textit{acknowledgement-based} and \textit{ternary}. In particular, we provide equilibrium characterizations for more than one channels, and give specific anonymous, equilibrium protocols with finite and infinite expected latency. In the equilibrium protocols with infinite expected latency, all players transmit successfully in optimal time, i.e. $\Theta(n/k)$, with probability tending to 1 as $n/k \to \infty$.

	\keywords{Contention resolution, Multiple channels, Acknowledgement-based protocol, Ternary feedback, Game theory}
\end{abstract}

\section{Introduction and Motivation}\label{Intro}

In the last fifteen years a great number of works in the Electrical and Electronics Engineering community has been devoted to designing medium access control (MAC) protocols that achieve high throughput. Their main approach is to consider, instead of the initial single-channel scheme, multi-channel schemes (\textit{multi-channel} MAC protocols) which resolve contention caused by packet collisions (e.g. \cite{MSW08,CSY03,SV04,SWM07,ZZHSS07,SLW15}). Apart from high throughput, an additional benefit of introducing more channels in such a system is robustness, meaning no great dependence on a single node's functionality. 
However, to the authors' knowledge, \textit{strategic} behaviour in multi-channel systems is limited to the Aloha protocol (\cite{MW03}), contrary to the case of single-channel systems (e.g. \cite{AAJ04,CLP14,FMN07,CGNRS16,CGNRS17}). In this paper, we examine the problem of \textit{strategic contention resolution} in multi-channel systems, where obedience to a suggested protocol is not required. We seek only \textit{anonymous}, equilibrium protocols, that is, protocols which do not use player IDs. If a player's protocol depended on her ID, then equilibria are simple, but can be unfair as well; scheduling each player's transmission through a priority queue according to her ID is an equilibrium.

We provide two types of equilibrium protocols. The first type, called \textit{FIN-EQ}, describes an anonymous, equilibrium protocol that yields finite expected time of successful transmission (\textit{latency}) to a player. Similarly, the second type, called \textit{IN-EQ}, describes an anonymous, equilibrium protocol which yields infinite expected latency to a player but is also \textit{efficient}, i.e, all players transmit successfully within $\Theta(\frac{\#players}{\#channels})$ time with high probability. We study equilibria for two classes of feedback protocols: (a) acknowledgement-based protocols, where the user gets just the information of whether she had a successful transmission or not, only when she tries to transmit her packet, and (b) protocols with ternary feedback, where the user is informed about the number of pending players in each time-step regardless of whether she attempted transmission or not. Previous results on these classes of protocols have been produced only for the case of a single transmission channel (\cite{CGNRS16,FMN07}). Here we investigate the multiple-channels case. 

In the last part of the paper we seek efficient protocols for both feedback classes. Due to an impossibility result that we show (Theorem \ref{thm: impos}), the technique used in \cite{FMN07} by Fiat et al. for the single-channel setting in order to provide a FIN-EQ that is also efficient, cannot be applied when there are more than one channels. This fact discourages us from searching for efficient FIN-EQ protocols and, instead, points to the search for efficient IN-EQ protocols, which indeed we find. One could argue that an anonymous protocol with infinite expected time until successful transmission, such as the IN-EQ protocols we provide, does not incentivize a player to participate in such a communication system. To this we reply that exponential waiting-time for a large amount of players (see protocol in Subsection \ref{ter: history-indep char}) is equally bad for a player, since waiting for e.g. $e^{10}$ msec is like waiting forever in Real-Time-Communications.  

\subsection{Our results}

The main contributions of this work are the characterizations of FIN-EQ and IN-EQ protocols in the two aforementioned feedback classes. Note that in the current bibliography regarding the single-channel setting, there are no characterizations of equilibrium in acknowledgement-based protocols. Also, in the single-channel setting the existence of a symmetric equilibrium with finite expected latency in the class of acknowledgement-based protocols remains an open problem, even for three players. However, for the settings with 2 and 3 transmission channels, we present simple anonymous FIN-EQ protocols for up to 4 and 5 players respectively which appeared in \cite{CMS18}. Furthermore, these protocols are memoryless, while the only known FIN-EQ protocol in the single-channel setting (\cite{CGNRS16}) is not.

The paper is organized in three main parts. Section \ref{section: ack} deals with FIN-EQ protocols in the acknowledgement-based feedback setting. In that section we give two characterizations of equilibrium and also provide FIN-EQ protocols for specific numbers of players and channels. Section \ref{section: ter} deals with FIN-EQ protocols in the ternary feedback setting and extends the corresponding results for the single-channel setting by Fiat et al. \cite{FMN07}. Finally, in Section \ref{eff prot}, IN-EQ protocols with deadline are provided with the property that the time until all $n$ players transmit successfully is $\Theta(n/k)$ with high probability, when there are $k$ channels. The latter result makes clear the advantage (with respect to time efficiency) that multiple channels bring to a system with strategic users, which is that the time until all players transmit successfully with high probability is inversely proportional to the number of available channels.

\subsection{Related work}

Contention in telecommunications is a major problem that results to poor throughput due to packet collisions. Motivated mainly by this problem, many works studying conflict-resolution protocols emerged in the late 70's (\cite{R75, C79, CJ79, H78, TM78}). Their approach is to resolve a collision when it occurs, and only then allow further transmissions on the channel. In those works the user's packets are assumed either to be generated by some stochastic process, or to appear at the same time in a worst-case scenario. Here, we consider the latter setting, i.e. a worst-case model of slotted time, where at any time-step all users have a packet ready to be transmitted (for an example of a similar bursty-input case, see \cite{BFHKL05}). As stated in \cite{LAG02}, even though real implementations of multiple-access channels do not fit precisely within the slotted-time model, it can be shown (e.g. \cite{KM87, ALR96}) that results obtained in this model do apply to realistic multiple-access channels.

Also, many works have examined multiple-channel communication protocols. In the data link layer, a Medium Access Control (MAC) protocol is responsible for the flow of data through a multiple-access medium. Our multiple-channels model is motivated by theoretical and experimental results which have shown that higher throughput and lower delay is achieved by using ``multi-channel'' MAC protocols (see \cite{NZD99,SWM07,MSW08,SV04}). In \cite{SV04},the \textit{multi-channel hidden terminal problem} is raised which, additionally to increased packet collisions, results to incapability of the users to ``sense'' more than one channels at a time (possibly none); therefore a user might not know whether another user transmitted successfully or not (see also \cite{TK75} for the classical ``hidden terminal problem''). This motivates us for the consideration of feedback protocols with minimum feedback, i.e. ``acknowledgement-based'' protocols (see par.2, Section \ref{Intro}). Also, settings with stronger feedback have been studied (e.g. the Aloha protocol in \cite{MW03}) in which a user is informed about the number of users that have not transmitted successfully yet. This is why we consider ``ternary feedback'' protocols (see par.2, Section \ref{Intro}).

Apart from the latter, all of the aforementioned works assume that the users blindly follow the given protocol, i.e. the users are not strategic. Contention resolution with strategic users has been studied only in single-channel settings or in the special case of the multiple-channel Aloha protocol. Some interesting cooperative and noncooperative models of slotted Aloha have been analysed in \cite{ABBA05,MMR06,MW03}. Aiming to understand the properties of contention resolution under selfishness, apart from various feedback settings, many cost functions have also been studied. One of the most meaningful cost functions is the one that models non-zero transmission costs as in \cite{CLP14} (and also \cite{AAJ04, MW03}). 

The theoretical works that relate the most to the current paper are the seminal paper by Fiat, Mansour and Nadav \cite{FMN07} and two by Christodoulou et al. \cite{CGNRS16, CGNRS17} which study protocols for strategic contention resolution with zero transmission costs. These works examine the case of a single transmission channel only. In \cite{FMN07} the feedback is ternary. In that work, a characterization of symmetric equilibrium is provided, along with an efficient FIN-EQ protocol that puts an extremely costly equilibrium after a deadline in order to force users to be obedient. The feedback model of \cite{CGNRS16} and \cite{CGNRS17} is the acknowledgement-based. Among other results, \cite{CGNRS16} provides the unique FIN-EQ protocol for the case of two players and a deadline IN-EQ protocol for at least three players.


\section{The Model and Definitions}\label{Model}

\paragraph{Game structure.} We define a \textit{contention game} as follows. Let $N = \{1,2,\dots,n\}$ be the set of players, also denoted by $[n]$, and $K=\{1,2, \dots, k\}$ the set of channels. Each player has a single packet that she wants to send through a channel in $K$, without caring about the identity of the channel. All players know $n$ and $K$. We assume synchronous communications with discretized time, i.e. time slots $t=1,2,\dots$. The players that have not yet successfully transmitted their packet are called \textit{pending} and initially all $n$ players are pending. At any given time slot $t$, a pending player $i$ has a set $A=\{0,1,2, \dots, k\}$ of \textit{pure strategies}: a pure strategy $a \in A$ is the action of choosing channel $a \in K$ to transmit her packet on, or no transmission ($a=0$). At time $t$, a \textit{(mixed) strategy} of a player $i$ is a probability distribution over $A$ that potentially depends on information that $i$ has gained from the process based on previous transmission attempts. If exactly one player transmits on a channel in a given slot $t$, then her transmission is \textit{successful}, the successful player exits the game (i.e. she is no longer pending), and the game continues with the rest of the players. On the other hand, whenever two or more players try to access the same channel (i.e. transmit) at the same time slot, a \textit{collision} occurs and their transmissions fail, in which case the players remain in the game. The game continues until all players have successfully transmitted their packets.

\paragraph{Transmission protocols.} Let $X_{i,t} \in A$ be the channel-indicator variable that keeps track of the identity of the channel where player $i$ attempted transmission at time $t$; value $0$ indicates no transmission attempt. For any $t \geq 1$, we denote by $\vv{X}_t$ the transmission vector at time $t$, i.e. $\vv{X}_t = (X_{1,t},X_{2,t}, \dots , X_{n,t})$. 

An \textit{acknowledgement-based} protocol uses very limited channel feedback. After each time step $t$, only players that attempted a transmission receive feedback, and the rest get no information. In fact, the information received by a player $i$ who transmitted during $t$ is whether her transmission was successful (in which case she gets an acknowledgement and exits the game) or whether there was a collision.

In a protocol with \textit{ternary feedback} every pending player in every round is informed about the number of remaining players $m \leq n$. This information is given to the players regardless of their transmission history.

Let $\vv{h}_{i,t}$ be the vector of the \textit{personal transmission history} of player $i$ up to time $t$, i.e. $\vv{h}_{i,t} = (X_{i,1},X_{i,2}, \dots , X_{i,t})$. We also denote by $\vv{h}_t$ the transmission history of all players up to time $t$, i.e. $\vv{h}_t = (\vv{h}_{1,t},\vv{h}_{2,t}, \dots \vv{h}_{n,t})$. A \textit{decision rule} $f_{i,t}$ for a pending player $i$ at time $t$, is a function that maps $\vv{h}_{i,t-1}$ to a strategy $\vv{P}_{i,t}$, with elements Pr$(X_{i,t}=a | \vv{h}_{i,t-1})$ for all $a \in A$. When the transmission probability on some $a' \in A$ is not stated in a decision rule it is because it can be deduced from the stated ones. 

For a player $i \in N$, a \textit{(transmission) protocol} $f_i$ is a sequence of decision rules $f_i = \{f_{i,t}\}_{t \geq 1} = f_{i,1}, f_{i,2}, \dots$. Given a protocol $f_i$ for player $i$, when her decision rules depend on the number of pending players and the personal history of $i$, then we describe them by the player's probability distribution on the action set $A$. In this case, we denote by $p_{m,t}^{i,a}$ the probability of player $i$ choosing action $a$ at time $t$ given her personal history $h_{t-1}$ when $m$ players are pending right before $t$. When the context is clear enough we will drop some of the indices accordingly. 

When we state that the players use an \textit{anonymous} protocol $f$, we will mean that they follow a common protocol $f$($=f_1 = \cdots = f_n$) whose decision rules do not depend on any ID of the player (in our setting players do not have IDs), i.e. the decision rule assigns the same strategy to all players with the same personal history. In particular, for any two players $i \neq j$ and any $t \geq 0$, if $\vv{h}_{i,t-1} = \vv{h}_{j,t-1}$, it holds that $f_{i,t}(\vv{h}_{i,t-1}) = f_{j,t}(\vv{h}_{j,t-1})$. In this case, we drop the subscript $i$ in the notation and write $f$ instead of $f_i$.


A protocol $f_i$ for player $i$ is a \textit{deadline protocol with deadline} $t_0$ if and only if there exists a finite $t_0 \geq 1$ such that a particular channel $a_i \in K$ is assigned (deterministically or stochastically) to player $i$ at some time $t \leq t_0$ and Pr$(X_{i,t}=a_i | \vv{h}_{i,t-1}) = 1$ for every time slot $t\geq t_0$ and any history $\vv{h}_{i,t-1}$.

\paragraph{Efficiency.} Assume that all $n$ players follow an anonymous protocol $f$. We will call $f$ \textit{efficient} if and only if all players will have successfully transmitted by time $\Theta(n/k)$ with high probability (i.e. with probability tending to 1, as $n \to \infty$).

\paragraph{Individual utility.} By \textit{protocol profile} $\vv{f} = (f_1 , f_2 , \dots , f_n)$ we will call the n-tuple of the players' protocols. For a given transmission sequence $\vv{X}_1, \vv{X}_2, \dots$, which is consistent with $\vv{f}$, define the \textit{latency} of agent $i$ as $T_i \triangleq \inf\{t : X_{i,t} = a, X_{j,t} \neq a, \text{ for some } a \in K, \forall j \neq i \}$. That is, $T_i$ is the time at which $i$ successfully transmits. 
Also, define the \textit{finishing time} of $\vv{f}$ as $T \triangleq \sup_i \{ T_i \}$, i.e., the least time at which all players have successfully transmitted.
Given a transmission history $\vv{h}_t$, the $n$-tuple of protocols $\vv{f}$ induces a probability distribution over sequences of further transmissions. In that case, we write $C_i^{\vv{f}}(\vv{h}_t) \triangleq \mathbb{E}[T_i | \vv{h}_t, \vv{f}] = \mathbb{E}[T_i | \vv{h}_{i,t}, \vv{f}]$ for the expected latency of a pending agent $i$ given that her current history is $\vv{h}_{i,t}$ and from $t+1$ on she follows $f_i$. For anonymous protocols, i.e. when $f_1 = f_2 = \cdots = f_n = f$, we will simply write $C_i^f(\vv{h}_t)$ instead. Abusing notation slightly, we will also write $C_i^{\vv{f}}(\vv{h}_0)$ for the \textit{unconditional} expected latency of player $i$ induced by $\vv{f}$. We also define the expected future latency $F_{i}^{\vv{f}}(\vv{h}_t) \triangleq C_i^{\vv{f}}(\vv{h}_t) - t$ and again, whenever clear from the context, we omit redundant indices or vectors from the notation.

\paragraph{Equilibria.} The objective of every player is to minimize her expected latency. 
We call a protocol $g_{i}$ a \textit{best response} of player $i$ to the \textit{partial protocol profile} $\vv{f}_{-i}$ if for any transmission history $\vv{h}_t$, player $i$ cannot decrease her expected latency by unilaterally deviating from $g_{i}$ after $t$. That is, for all time slots $t$, and for all protocols $f_{i}'$ for player $i$, we have
\begin{align*}
C_i^{(\vv{f}_{-i}, g_i)}(\vv{h}_t) \leq C_i^{(\vv{f}_{-i}, f'_i)}(\vv{h}_t),
\end{align*}
where $(\vv{f}_{-i}, g_i)$ (respectively, $(\vv{f}_{-i}, f'_i)$) denotes the \textit{protocol profile} where every player $j \neq i$ uses protocol $f_j$ and player $i$ uses protocol $g_i$ (respectively $f'_i$). For an anonymous protocol $f$, we denote by $(f_{-i},g_i)$ the profile where player $j \neq i$ uses protocol $f$ and player $i$ uses protocol $g_i$.

We say that $\vv{f} = (f_1, f_2, \dots f_n)$ is an \textit{equilibrium} if for any transmission history $\vv{h}_t$ the players cannot decrease their expected latency by unilaterally deviating after $t$; that is, for every player $i$, $f_i$ is a best response to $\vv{f}_{-i}$. 

\paragraph{FIN-EQ and IN-EQ protocols.} We call an anonymous protocol \textit{FIN-EQ} if it is an equilibrium protocol and yields finite expected latency to a player. Similarly, we call an anonymous protocol \textit{IN-EQ} if it is an equilibrium protocol, yields infinite expected latency to a player, and is also efficient.

\section{Equilibrium for Acknowledgement-based Protocols}\label{section: ack}

\subsection{Nash equilibrium characterizations}\label{ack-based characterizations}

The following equilibrium characterizations for the class of acknowledgement-based protocols help us check whether the protocols we subsequently guess are equilibrium protocols. The characterizations are for symmetric and asymmetric equilibria, arbitrary number of channels $k \geq 1$ and number of players $n \geq 2$.

In an acknowledgement-based protocol, the actions of player $i$ at time $t$ depend only (a) on her personal history $\vv{h}_{i, t-1}$ and (b) on whether she is pending or not at $t$. 
Let $\vv{f}=(f_1, f_2, \dots, f_n)$ be a tuple of acknowledgement-based protocols (not necessarily anonymous) for the $n$ players.
For a (finite) positive integer $\tau^*$, and a given history $h_{i,\tau^*}=(a_{i,1}, a_{i,2}, \dots, a_{i,\tau^*})$, define for player $i$ the protocol
\begin{align}\label{prot: g}
g_i = g_{i}(h_{i,\tau^*}) \triangleq 
\begin{cases}  
& \left( \text{Pr}\{X_{i,t} = a_{i,t} \} = 1, \text{ } \text{Pr}\{X_{i,t} \neq a_{i,t} \} = 0 \right) \quad \text{, for } 1 \leq t \leq \tau^* \\
& f_{i,t}, \quad \text{for } t > \tau^* .  
\end{cases}
\end{align}

A personal history $\vv{h}_{i,\tau^*}$ is \textit{consistent with} the protocol profile $\vv{f}$ if and only if there is a non-zero probability that $\vv{h}_{i,\tau^*}$ will occur for player $i$ under $\vv{f}$.
Protocol $g_{i}(h_{i,\tau^*})$ is \textit{consistent} with $\vv{f}$ if and only if $h_{i,\tau^*}$ is consistent with $\vv{f}$, and when clear from the context we write $g_i$ instead. We denote the set of all $g_i$'s, that is, all $g_{i}(h_{i,t})$'s for all $t \geq 1$, which are consistent with $\vv{f}$, by $\mathcal{G}_{i}^{\vv{f}}$. If $f_i=f$ $\forall i$ (i.e. $f$ is anonymous), then instead of $g_i$ and $\mathcal{G}_{i}^{\vv{f}}$ we write $g$ and $\mathcal{G}^{f}$ respectively. 

\begin{lemma}[Equilibrium characterization 1]\label{eq char 1}
	Consider a profile ~\\ $\vv{f} = (f_1, f_2, \dots f_n)$ of acknowledgement-based protocols and a protocol ~\\ $g_i = g_{i}(h_{i,\tau^*})$ for some $\tau^* \geq 1$. The following statements are equivalent:\\
	(i) $\vv{f}$ is an equilibrium. \\
	(ii) For every player $i \in [n]$, \\
	\text{\quad } if $g_i \in \mathcal{G}_{i}^{\vv{f}}$ then ${C_{i}^{(\vv{f}_{-i},g_i)}(\vv{h}_0) = \min\limits_{f_{i}'} C_{i}^{(\vv{f}_{-i},f_{i}')}(\vv{h}_0) = C_{i}^{\vv{f}}(\vv{h}_0)}$.
\end{lemma} 

\begin{proof}
	
	To show that $\vv{f}$ being an equilibrium is a sufficient condition, we use the same argument as in Lemma 4 of \cite{CGNRS16}. In particular, for a player $i$, due to the Tower Property we have,
	\begin{align}\label{expected 1}
	C_{i}^{\vv{f}}(\vv{h}_0) &= \mathbb{E}[T_i | \vv{h}_{i,0},\vv{f}] \nonumber \\
	&= \sum_{\vv{h}_{i,\tau^*}} \mathbb{E}[T_i | \vv{h}_{i,0},(\vv{f}_{-i},g_{i}(h_{i,\tau^*}))] \text{Pr}\{\vv{h}_{i,\tau^*} \text{ happens for } i \} .
	\end{align}
	
	For short, we will denote $g_{i}(h_{i,\tau^*})$ by $g_i$, thus we denote $\mathbb{E}[T_i | \vv{h}_{i,0},(\vv{f}_{-i},g_{i}(h_{i,\tau^*}))]$ by $C_{i}^{(\vv{f}_{-i},g_{i})}(\vv{h}_0)$.
	%
	%
	Then, suppose that $\vv{f}$ is an equilibrium and assume for the sake of contradiction that there is a transmission history $\vv{h}_{i,\tau*}$ for player $i$ such that $C_{i}^{(\vv{f}_{-i},g_{i})}(\vv{h}_0) \neq C_{i}^{\vv{f}}(\vv{h}_0)$. Obviously, if $C_{i}^{(\vv{f}_{-i},g_{i})}(\vv{h}_0) < C_{i}^{\vv{f}}(\vv{h}_0)$ this would mean that protocol $g_i(\tau^*)$ is better than $f_i$, thus $\vv{f}$ is not an equilibrium. If, on the other hand, $C_{i}^{(\vv{f}_{-i},g_{i})}(\vv{h}_0) > C_{i}^{\vv{f}}(\vv{h}_0)$, then from (\ref{expected 1}) there must exist another transmission history $\vv{h}_{i,\tau^*}'$ such that $C_{i}^{(\vv{f}_{-i},g_{i}(\vv{h}_{i,\tau^*}'))}(\vv{h}_0) < C_{i}^{\vv{f}}(\vv{h}_0)$. Therefore, we conclude that $C_{i}^{\vv{f}}(\vv{h}_0) = C_{i}^{(\vv{f}_{-i},g_{i})}(\vv{h}_0)$ which also equals $\min\limits_{f_{i}'} C_{i}^{(\vv{f}_{-i},f_{i}')}(\vv{h}_0)$ by definition of the equilibrium, for every transmission history $\vv{h}_{i,\tau^*}$ that is consistent with $\vv{f}$.
	
	To show that $\vv{f}$ being an equilibrium is also a necessary condition, assume that $g_i \in \mathcal{G}_{i}^{\vv{f}}$ implies $C_{i}^{(\vv{f}_{-i},g_i)}(\vv{h}_0) = \min\limits_{f_{i}'} C_{i}^{(\vv{f}_{-i},f_{i}')}(\vv{h}_0)$. Then, equality (\ref{expected 1}) becomes
	\begin{align*}
	C_{i}^{\vv{f}}(\vv{h}_0) & = \sum_{\vv{h}_{i,\tau^*}} C_{i}^{(\vv{f}_{-i},g_{i}(h_{i,\tau^*}))}(\vv{h}_0) \text{Pr}\{\vv{h}_{i,\tau^*} \text{ happens for } i \} \\
	& = \sum_{\vv{h}_{i,\tau^*}} \min\limits_{f_{i}'} C_{i}^{(\vv{f}_{-i},f_{i}')}(\vv{h}_0) \text{Pr}\{\vv{h}_{i,\tau^*} \text{ happens for } i \} \\
	& = \min\limits_{f_{i}'} C_{i}^{(\vv{f}_{-i},f_{i}')}(\vv{h}_0)
	\end{align*}
	and thus $\vv{f}$ is by definition an equilibrium.
\end{proof}

\begin{corollary}[Best response]\label{best response}
	Consider a profile $\vv{f} = (f_1, f_2, \dots f_n)$ of acknowledgement-based protocols. For a fixed protocol $f_{i}'$ of player $i \in [n]$ and some $h_{i,\tau^*}=(a_{i,1}, a_{i,2}, \dots, a_{i,\tau^*})$ consistent with $(\vv{f}_{-i},f_{i}')$, define the following protocol. 
	\begin{align}\label{prot: r}
	r_i = r_{i}(h_{i,\tau^*}) \triangleq 
	\begin{cases}  
	& \left( \text{Pr}\{X_{i,t} = a_{i,t} \} = 1, \quad \text{Pr}\{X_{i,t} \neq a_{i,t} \} = 0 \right) \quad \text{, for } 1 \leq t \leq \tau^* \\
	& f_{i,t}', \quad \text{, for } t > \tau^*. 
	\end{cases}
	\end{align}
	If for player $i$ there exists a finite $\tau^* \geq 1$ such that $C_{i}^{(\vv{f}_{-i},r_{i}(h_{i,\tau^*}))}(\vv{h}_0) \geq C_{i}^{(\vv{f}_{-i},f_{i})}(\vv{h}_0)$ for every $h_{i,\tau^*}$,
	then $C_{i}^{(\vv{f}_{-i},f_{i}')}(\vv{h}_0) \geq C_{i}^{(\vv{f}_{-i},f_{i})}(\vv{h}_0)$.
\end{corollary}

\begin{proof}
	By definition of the expected latency (equation (\ref{expected 1})) for a fixed $\tau^*$ we have:
	\begin{align*}
	C_{i}^{(\vv{f}_{-i},f_{i}')}(\vv{h}_0) &= \sum_{\vv{h}_{i,\tau^*}} C_{i}^{(\vv{f}_{-i},r_{i}(h_{i,\tau^*}))}(\vv{h}_0) \text{Pr}\{\vv{h}_{i,\tau^*} \text{ happens for } i \} \\ 
	&\geq \sum_{\vv{h}_{i,\tau^*}} C_{i}^{(\vv{f}_{-i},f_{i})}(\vv{h}_0) \text{Pr}\{\vv{h}_{i,\tau^*} \text{ happens for } i \} \\
	&=C_{i}^{(\vv{f}_{-i},f_{i})}(\vv{h}_0).
	\end{align*}
\end{proof}


\begin{lemma}[Equilibrium characterization 2]\label{eq char 2}
	Consider a profile ~\\ $\vv{f} = (f_1, f_2, \dots f_n)$ of acknowledgement-based protocols. The following statements are equivalent: \\
	(i) $\vv{f}$ is an equilibrium. \\
	(ii) For every player $i \in [n], \\
	\text{\quad } \begin{cases}  
	\text{(a)} \quad	C_{i}^{(\vv{f}_{-i},g_i)}(\vv{h}_0) = C_{i}^{(\vv{f}_{-i},r_i)}(\vv{h}_0) = C_{i}^{\vv{f}}(\vv{h}_0), \quad \forall g_i, r_i \in \mathcal{G}_{i}^{\vv{f}}, \text{ and} \\
	\text{(b)} \quad  C_{i}^{(\vv{f}_{-i},g_i)}(\vv{h}_0) \leq C_{i}^{(\vv{f}_{-i},r_i)}(\vv{h}_0), \quad \forall g_i \in \mathcal{G}_{i}^{\vv{f}}, r_i \notin \mathcal{G}_{i}^{\vv{f}}. 
	\end{cases}	
	$		
\end{lemma}

\begin{proof}
	Sufficiency of $\vv{f}$ being an equilibrium for condition (ii-a) comes directly from Lemma \ref{eq char 1}; for condition (ii-b), for the sake of contradiction suppose $\vv{f}$ is an equilibrium and that there exist some protocols $g_i \in \mathcal{G}_{i}^{\vv{f}}$ and $ r_i \notin \mathcal{G}_{i}^{\vv{f}}$ such that $C_{i}^{(\vv{f}_{-i},g_i)}(\vv{h}_0) > C_{i}^{(\vv{f}_{-i},r_i)}(\vv{h}_0)$. This means that $r_i$ is a better protocol than $f_i$, thus $(\vv{f}_{-i},f_i)$ is not an equilibrium, which is a contradiction.
	
	To prove necessity of $\vv{f}$ being an equilibrium under conditions (ii-a) and (ii-b), for the sake of contradiction, suppose (ii-a) and (ii-b) hold and $\vv{f}$ is not an equilibrium. Then there must exist some protocol $f_{i}'$ such that $C_{i}^{(\vv{f}_{-i},f_{i}')}(\vv{h}_0) < C_{i}^{\vv{f}}(\vv{h}_0)$. Using (\ref{expected 1}) the latter inequality can be written as
	\begin{align*}
	\sum_{\vv{h}_{i,\tau^*}} C_{i}^{(\vv{f}_{-i},r_{i}(h_{i,\tau^*}))}(\vv{h}_0) \text{Pr}\{\vv{h}_{i,\tau^*} \text{ happens for } i \} < \sum_{\vv{h}_{i,\tau^*}} C_{i}^{(\vv{f}_{-i},g_{i}(h_{i,\tau^*}))}(\vv{h}_0) \text{Pr}\{\vv{h}_{i,\tau^*} \text{ happens for } i \},
	\end{align*}
	where $g_{i}(h_{i,\tau^*})$ is consistent with $\vv{f}$ and $r_{i}(h_{i,\tau^*})$ is consistent with $(\vv{f}_{-i}, f_{i}')$. Given the conditions (ii-a) and (ii-b) the latter inequality is a contradiction.
\end{proof}

%


\subsection{Acknowledgment-based FIN-EQ protocols}\label{ack: FIN-EQ}

Regarding the search for FIN-EQ protocols, there is no straight-forward way for our equilibrium characterizations (previous subsection) to be used in order to \textit{find} an equilibrium protocol. However, they allow us to \textit{check} whether the protocols discussed in this subsection are equilibrium protocols. In this subsection we give FIN-EQ protocols for $k=2$ and $k=3$.

We define the following anonymous, memoryless protocol for $k \geq 2$ channels.

\begin{tcolorbox}[ams gather]
	\textbf{\underline{Protocol $\mathbf{\emph{f}^{\ k}}$}}: \text{For player $i$, every $t \geq 1$ and any history $\vv{h}_{i,t-1}$,} \nonumber \\
	f_{i,t}^k = \left(\text{Pr}\{X_{i,t} = 0 \} = 0, \quad \text{Pr}\{X_{i,t} = a \} = \frac{1}{k}, \quad \forall a \in K \right). \label{protocol(n,2)}
\end{tcolorbox}

\paragraph{\textbf{n players - 2 transmission channels.}}
Here, we first give an example of a method for checking equilibria (Theorem \ref{3 pl - 2 ch}). Then, with a better approach, by employing our characterizations of the previous subsection, we prove that $f^2$ is an equilibrium protocol for $n \in \{2,3,4\}$ players and $k=2$ channels (Theorem \ref{thm: eq (n,2)}). 

%
%
%

%

\begin{lemma}\label{exp latency (n,2)}
	When all $n \geq 2$ players use protocol $f^2$ the expected latency of any player is $2^{n}/n$.
\end{lemma}

\begin{proof}
	The process from the perspective of an arbitrary player $i$ can be modelled as the following Markov chain; the states are named after the number of remaining players including $i$, and state $\langle \times \rangle$ is the state where $i$ finds herself after successful transmission.
	
	We write $p_{x}^{y}$ to denote the transition probability to go from state $\langle x \rangle$ to state $\langle y \rangle$. We have
	\begin{equation}\label{eq (n,2) 1}
	\begin{rcases*}
	p_{m}^{\times} = \left(\frac{1}{2}\right)^{m-1} \\
	p_{m}^{m-1} = (m-1)\left(\frac{1}{2}\right)^{m-1} \\
	p_{m}^{m} = 1 - m \left(\frac{1}{2}\right)^{m-1}
	\end{rcases*}
	\forall 3 \leq m \leq n \text{ , and }
	\end{equation}
	\begin{align}\label{eq (n,2) 2}
	p_{2}^{\times} = \frac{1}{2}, \quad p_{2}^{2} = \frac{1}{2}.
	\end{align}
	
	The expected absorption time from state $\langle n \rangle$ to state $\langle \times \rangle$ is found from the following set of equations:
	\begin{align*}
	& h_{m}^{\times} = 1 + p_{m}^{m} h_{m}^{\times} + p_{m}^{m-1} h_{m-1}^{\times} , \quad \text{for all } 3 \leq m \leq n,  \\
	\text{and} \quad & h_{2}^{\times} = 1 + p_{2}^{2} h_{2}^{\times} ,
	\end{align*}
	where $h_{x}^{y}$ denotes the expected hitting time from state $\langle x \rangle$ to state $\langle y \rangle$. By solving this system of linear equations we get
	\begin{align*}
	h_{n}^{\times} = \frac{2^n}{n}, \quad \text{for } n\geq 2. 
	\end{align*}
\end{proof}

In the next theorem we will give an example of a method for checking whether a given protocol profile is an equilibrium, which however could be inconclusive in some cases. Suppose you we want to check whether an arbitrary protocol profile $\vv{f}$ is an equilibrium. By definition of the equilibrium, we can fix all protocols except player $i$'s, i.e. $\vv{f}_{-i}$ and check if $f_i$ is a best response to them, and repeat this for every player $i$. By fixing $\vv{f}_{-i}$ we create a stochastic environment for player $i$ who can be considered to be free to take sequential decisions through time. These decisions correspond to decision rules of $f_{i}$. Since, due to the feedback limitations, $i$ has no information about the number of pending players, this situation from her point of view is modeled as an infinite state Partially Observable Markov Decision Process (POMDP). $f_i$ is a best response to $\vv{f}_{-i}$ if and only if $f_i$ is an optimal policy of the POMDP, that is, a set of decisions through time that minimize her expected latency. 

However for this kind of POMDPs there are no known techniques to find an optimal policy. In order to circumvent this problem, we can assume that player $i$ is an advantageous player that always knows how many players are pending. This turns the infinite state POMDP into a finite state Markov Decision Process (MDP), whose optimal policy we can find through known techniques (e.g. \cite{N98}). One can see that the optimal policy in the MDP of the advantageous player $i$ yields at most the expected latency of the optimal policy in the POMDP of the initial player $i$. Thus, if the best policy in the MDP yields the same expected latency as what $\vv{f}$ gives to $i$, then we know that $f_i$ is a best response; however, if the best policy of the MDP yields smaller expected latency, then we get no information about whether $f_i$ is a best response in the POMDP or not. The proof of the next theorem demonstrates the method and shows that protocol $f^2$ of (\ref{protocol(n,2)}) is an equilibrium protocol for 3 players.

\begin{theorem}\label{3 pl - 2 ch}
	For 3 players and 2 channels, $f^2$ is an equilibrium protocol with expected latency $8/3$.
\end{theorem}

\begin{proof}
	Consider the Markov Decision Process (MDP) $(T,S_t, A_{s,t}, p_{t}(j | s,a), r_{t}(s,a))$, where $S_t$ is the state space for time $t$; $A_{s,t}$ is the set of possible actions that can be taken after observing state $s$ at time $t$; $p_{t}(j | s,a)$ defines the transition probability to state $j \in S_{t+1}$ at time $t+1$, and only depends on the state $s$ and chosen action $a$ at time $t$; $r_{t}(s,a)$ is the cost function that determines the immediate cost for the agent's choice of action $a$ while in state $s$. When the state $s$ cannot be observed with certainty at time $t$, the agent only knows a probability distribution, called \textit{belief state}, over $S_t$. The process then is called Partially Observable Markov Decision Process (POMDP). An \textit{optimal policy} $\pi: S \to A$ is a function that rules, for each state or belief state, which action to perform, with an objective to minimize the expected cost.
	
	For the proof of the above theorem we will use the following property of POMDPs. This property comes directly from the fact that an agent optimizing over all policies that every time consider her exact state gets a better policy than an agent that knows a probability distribution on the state space (belief states).  
	\begin{proposition}\label{POMDP-MDP}
		An optimal policy $\pi_1$ of an agent in a POMDP yields as expected cost at least the expected cost of the optimal policy $\pi_2$ of the corresponding MDP, in which at any time $t$ the agent observes her exact state.
	\end{proposition}
	%
	
	To prove Theorem \ref{3 pl - 2 ch} we think as follows. Let us fix protocol $f^2$ as defined in (\ref{protocol(n,2)}) for two players, and let the remaining player $i$ have an arbitrary protocol $g_{i}$. Then let us find the optimal policy for $i$. If and only if the optimal policy yields expected cost strictly lower than what protocol $f^2$ would yield for player $i$ (due to Lemma \ref{exp latency (n,2)}, that is $8/3$), then $f^2$ is not an equilibrium protocol. 
	The game stated at Theorem \ref{3 pl - 2 ch}, from player $i$'s perspective, is modelled by a POMDP where each state is determined by the number of pending players, with an additional absorbing state - where $i$ goes after successfully transmitting - and $i$'s transmission history for every $t \geq 1$. Player $i$'s belief state at any time $t$ is determined by her belief state at time $t-1$, the action she chose at time $t-1$, and her observation (e.g. her transmission history up to $t-1$). This is a POMDP with infinite states, for which, to the best of our knowledge, currently there are no methods in the literature for finding an optimal policy.
	
	However, we will find the best policy and the expected cost of the corresponding MDP, where player $i$ knows in what state she finds herself after an action and observation. This expected cost is a lower bound on the expected cost of the optimal policy of the original POMDP (see Proposition \ref{POMDP-MDP}). In the MDP we create, player $i$ knows at any time $t$ how many players are pending and her transmission history up to time $t$.

	
	Let $p \in \{1,2,3\}$ indicate the number of pending players. Observe that the time steps at which the process has a given $p$ are consecutive; without loss of generality assume that for some $p$, the process is in the discrete time interval $[\tau_{p} , \tau_{p-1} -1]$, where we set $\tau_3 = 1$. Consider now the set $S_p$ of all states $s_{p}(\vv{h}_{i,t})$ of the MDP, where the number of pending players $p \in \{1,2,3\}$ is fixed, whereas the transmission history $\vv{h}_{i,t}$ for $\tau_{p} \leq t < \tau_{p-1}$ can vary. Because of the protocol $f$ being memoryless, the same action (probability distribution over action space $A$) of $i$ chosen at any state in $S_p$ produces the same transition probabilities. Therefore, choosing the optimal policy makes the set $S_p$ of states collapse to a single state $s_p$, where $p \in \{1,2,3\}$. The resulting MDP is a finite MDP with states $s_1, s_2, s_3$ and $s_{\times}$, where the latter is an absorption state to which player $i$ goes after a successful transmission. Denote the expected cost of the MDP's optimal policy given that the initial state is $s_p$ by $c(s_p)$. In our problem the immediate cost for any combination of state and action is 1, since we count the number of rounds in which $i$ is pending. Using Lemma 5.4.2 and Theorem 5.4.3 of \cite{N98} we can find $c(s_3)$ by solving the following system of linear equations
	\begin{align}\label{best_policy1}
	c(s_p) = 1 + \sum_{s' \in \{s_1,s_2,s_3\}} \text{Pr}(s_p \text{ to } s' | \text{ policy } \pi) c(s') .
	\end{align}
	Then, by minimizing each $c(s_p)$ over policies $\pi$ we get the optimal expected costs $C(s_p)$, $p \in\ \{1,2,3\}$. As a byproduct of the minimization we find the best policy $\pi^*$.
	
	In our problem, a policy $\pi$ is a tuple ($q_1, z_1, q_2, z_2, q_3, z_3$), where $q_p$, $p \in \{1,2,3\}$ determines the probability that player $i$ will attempt a transmission, and $z_p$, $p \in \{1,2,3\}$ determines the probability that she will attempt the transmission on channel $a=1$. To give a small example, for a given state $s_p$, $(\text{Pr}(X_t = 0), \text{Pr}(X_t = 1), \text{Pr}(X_t = 2))=(1-q_p, q_p z_p, q_p (1-z_p))$. By solving system (\ref{best_policy1}), we get that 
	\begin{align*}
	c(s_1) = \frac{1}{q_1}, \quad c(s_2)=\frac{2+2q_1-2q_2}{2q_1-q_1 q_2}, \quad c(s_3)=2+\frac{4-2q_2 - 2q_2 q_3 + 2q_1 q_2 q_3}{4q_1 - 2q_1 q_2 + 2q_1 q_3 - q_1 q_2 q_3} 
	\end{align*}
	which implies that a policy does not depend on any of the $z_p$'s. Now, by minimizing the above expected costs we get $C(s_1)=1, C(s_2)=2$ and $C(s_3)=8/3$ for $q_1=1$ and $q_3=1$. Note that the optimal policy allows $z_1,z_2,z_3$ and $q_2$ to be arbitrary probabilities. $q_2$ being even 0 is not a contradiction since in our MDP the player is always aware of the pending players (state); in the case where $q_2=0$, when the player is in state $s_2$, she waits one round until the other player transmits successfully and then realizes that she is alone pending in $s_1$; in the next round she transmits with probability 1.  
	
	We have shown that a best policy of an advantageous player gives her the same expected latency as protocol $f^2$ defined in (\ref{protocol(n,2)}) (the expected latency of $f^2$ is given by Lemma \ref{exp latency (n,2)}). This, combined with Proposition \ref{POMDP-MDP} completes the proof of Theorem \ref{3 pl - 2 ch}.
\end{proof}

We subsequently exploit the lack of memory and the anonymity of our protocol $f^2$ defined in equation (\ref{protocol(n,2)}) and show more general results on equilibria (Theorem \ref{thm: eq (n,2)}), using the characterizations from Subsection \ref{ack-based characterizations}.

\begin{theorem}\label{impossibility 1}
	In a contention game with $k=2$ channels, consider an anonymous, memoryless protocol of player $i$ with the property: Pr$\{ X_{i,t} = 0 \} = 0$, for every $t \geq 1$. For more than 4 players any such protocol is not an equilibrium protocol.
\end{theorem}

\begin{proof}
	Assume that an anonymous protocol $f$ as stated in the theorem is an equilibrium protocol for $n \geq 5$ players. We will show that condition (ii-b) of Lemma \ref{eq char 2} does not hold. That is, if $n \geq 5$ players use a protocol $f$ with the property that in each time its decision rule assigns zero probability to ``no transmission'', then there exists a best response that yields strictly better expected latency for an arbitrary player.
	
	Suppose $f$ is an equilibrium protocol. $f$ consists of a decision rule for each time slot $t$, i.e. a probability distribution on the available channels (with probability 0 of ``no transmission'' as the theorem's statement requires). Since all players use this protocol, in an arbitrary time $t$ all players have the same distribution on the channels. For the sake of contradiction, suppose there is some $t'$ for which the decision rule is other than $\left(\text{Pr}\{X_{i,t} = 1 \} = \frac{1}{2}, \quad \text{Pr}\{X_{i,t} = 2 \} = \frac{1}{2} \right)$. Without loss of generality, we have $\text{Pr}\{X_{i,t} = 1 \} > \text{Pr}\{X_{i,t} = 2 \}$. Thus, an arbitrary player $i$, at time $t$, can unilaterally change her distribution to $\left(\text{Pr}\{X_{i,t} = 1 \} = 0, \quad \text{Pr}\{X_{i,t} = 2 \} = 1 \right)$ and increase her probability of transmitting successfully in the specific round. As a consequence her expected latency would strictly decrease, hence a protocol with a decision rule with different probabilities on each channel cannot be in a symmetric equilibrium. Therefore, the anonymous, equilibrium protocol $f$, with the property Pr$\{ X_{i,t} = 0 \} = 0$ for every $t \geq 1$, prescribes $\left(\text{Pr}\{X_{i,t} = 1 \} = \frac{1}{2}, \quad \text{Pr}\{X_{i,t} = 2 \} = \frac{1}{2} \right)$ for every $t \geq 1$. The expected latency of a player using such a protocol, when there are $n$ pending players, is found in Lemma \ref{exp latency (n,2)} to be $2^n / n$.
	
	We will show that, when the number of pending players at $t=0$ is $n \geq 5$, protocol
	\begin{align*}
	g_i \triangleq 
	\begin{cases}  
	& \left( \text{Pr}\{X_{i,1} = 1 \} = 0, \quad \text{Pr}\{X_{i,1} = 2 \} = 0 \right)  \\
	& \left( \text{Pr}\{X_{i,t} = 1 \} = \frac{1}{2}, \quad \text{Pr}\{X_{i,t} = 2 \} = \frac{1}{2} \right), \quad \text{for } t \geq 2 ,  
	\end{cases}
	\end{align*}
	is a better response for an arbitrary player $i$, that is, $C_{i}^{(f_{-i},g_{i})}(\vv{h}_{i,0}) < C_{i}^{f}(\vv{h}_{i,0}) = 2^n / n$. 
	
	Suppose player $i$ uses protocol $g_i$ when there are $n \geq 5$ pending players at $t=0$. At time $t=2$ she is not aware of the number of players that remain pending. However, there are two cases, either $n$ players are pending in case none of the other $n-1$ players in $t=1$ transmitted successfully, or $n-1$ players remain in case only one of the other $n-1$ players transmitted successfully in $t=1$. Note that there is no way that two players cannot simultaneously transmit successfully in round $t=2$ due to the given protocol $f$ and the number of pending players. The probability for each of the two aforementioned events is,
	\begin{align*}
	P_{n-1}(x) = \sum_{r=x}^{n-1} (-1)^{r-x} \binom{r}{x} \binom{2}{r}\binom{n-1}{r}r!\left(\frac{1}{2}\right)^r \left(1 - \frac{r}{2}\right)^{n-1-r}
	\end{align*}	
	where $x$ is the number of players that transmit successfully, $0^0 \triangleq 1$, and $\binom{a}{b} \triangleq 0$ for $a < b$. To see how this formula is produced, please refer to the proof of Lemma \ref{balls-bins} (Section \ref{eff prot}), up to equation (\ref{prob P_n(x)}). Here, equation (\ref{prob P_n(x)}) is used for $z=1$ and $k=2$.
	
	In order to capture the dependence of the expected future cost (after history $h_{t-1}$) on the number of pending players $n$, when player $i$ uses $g_i$ and the rest of the players use $f$, we denote it by $F_{i,n}^{(\vv{f}_{-i},g_{i})}(\vv{h}_{i,t - 1})$. Similarly, we denote the expected latency by $C_{i,n}^{(\vv{f}_{-i},g_{i})}(\vv{h}_{i,t - 1})$. We have,

	\begin{align}\label{equation: 6}
	C_{i,n}^{(f_{-i},g_{i})}(\vv{h}_{i,0}) = F_{i,n}^{(f_{-i},g_{i})}(\vv{h}_{i,0}) = 1 &+ P_{n-1}(0)F_{i,n}^{(f_{-i},g_{i})}(\vv{h}_{i,1}) + P_{n-1}(1)F_{i,n-1}^{(f_{-i},g_{i})}(\vv{h}_{i,1})  \nonumber \\
	= 1 &+ P_{n-1}(0) \frac{2^n}{n} + P_{n-1}(1) \frac{2^{n-1}}{n-1}.
	\end{align}
	
	For $n \geq 5$, our formula in (\ref{prob P_n(x)}) gives $P_{n-1}(0) = 1 - (n-1)\left(\frac{1}{2}\right)^{n-2}$ and $P_{n-1}(1) = (n-1)\left(\frac{1}{2}\right)^{n-2}$. Therefore (\ref{equation: 6}) becomes
	\begin{align*}
	C_{i,n}^{(f_{-i},g_{i})}(\vv{h}_{i,0}) & = 1 + \left[ 1 - (n-1)\left(\frac{1}{2}\right)^{n-2} \right] \frac{2^n}{n} + (n-1)\left(\frac{1}{2}\right)^{n-2} \frac{2^{n-1}}{n-1}   \\
	& = \frac{2^n}{n} + \frac{4}{n} - 1  \\
	& < \frac{2^n}{n} \quad , \text{ for } n > 4 \\
	& = C_{i,n}^{f}(\vv{h}_{i,0}).
	\end{align*}
	Thus protocol $g_i$ yields strictly smaller expected latency than $f_i$ for player $i$ when $n \geq 5$, and this means that $f$ is not a symmetric equilibrium for $n \geq 5$.  
\end{proof}

Since protocol $f^2$ belongs to the class of protocols defined in the statement of Theorem \ref{impossibility 1}, the following corollary is immediate.

\begin{corollary}
	For $n \geq 5$ players and $k=2$ channels, $f^2$ is not an equilibrium protocol. In fact, a better response for any player is to not transmit in $t=1$ and then follow $f^2$.
\end{corollary}

Now we prove two lemmata that, combined with our second characterization of equilibria (Lemma \ref{eq char 2}), result to one of this section's main theorems (Theorem \ref{thm: eq (n,2)}) that determines equilibrium protocols for $n \in \{2,3,4\}$ players and $k=2$ channels. In particular, we will show that for number of players $n=2$, $n=3$ and $n=4$, when $n-1$ players use $f$, if some deviator unilaterally chooses any possible protocol $g_i$ as defined in (\ref{prot: g}) that \textit{is} consistent with $\vv{f}$, she will suffer the same expected latency, namely $2^n / n$. Then, we will show that if the deviator unilaterally chooses any possible protocol as defined in (\ref{prot: g}) that \textit{is not} consistent with $\vv{f}$, she will suffer expected latency at least $2^n / n$. These two facts, by Lemma \ref{eq char 2}, show that $f$ is an equilibrium protocol for $n \in \{2,3,4\}$.

\begin{lemma}\label{cond ii-a}
	For $n \geq 2$ players and $k=2$ channels, any player $i$ that follows a protocol $g_i \in \mathcal{G}^{f^2}$ in the profile $(f_{-i}^2,g_i)$, where $f^2$ is defined in (\ref{protocol(n,2)}), has expected latency $2^n / n$.
\end{lemma}

\begin{proof}
	%
	Consider the contention game with fixed number of players $n \geq 2$ and 2 channels. $n-1$ players use protocol $f^2$ and a player $i \in [n]$ uses some protocol $g_{i}(h_{i,\tau^*}) \in \mathcal{G}^{f^2}$ as defined in (\ref{prot: g}), for some $\tau^* \geq 1$. To make easier our reference to the expected future latency of a player in the special case where (almost) all players follow protocol $f^2$ of (\ref{protocol(n,2)}), and to capture the number of players in the notation, we will denote by $D(r_{i},n) \triangleq \mathbb{E}[T_i | \vv{h}_{i,0},(\vv{f}_{-i}^2,r_{i})]$ and $D(f_{i}^2,n) \triangleq \mathbb{E}[T_i | \vv{h}_{i,0},\vv{f^2}]$ the expected future latency of player $i$ when $n$ players participate.
	
	First we show that condition (ii-a) of Lemma \ref{eq char 2} holds for every $n \geq 2$. From Lemma \ref{exp latency (n,2)} we know that $D(f_{i}^2,n) = 2^n/n$, for every $i \in [n]$. Now observe that the set of all protocols $g_i(\tau^*)$ as defined in \ref{prot: g} that are consistent with $f_i^2$, consists of the protocols for which $a_t \neq 0$ for every $1 \leq t \leq \tau^*$ for any $\tau^* \geq 1$. That is, for all possible tuples $(a_1, a_2,\dots, a_{\tau^*})$ of a given $\tau^*$, there is no $t \leq \tau^*$ for which $a_t = 0$, and this is for all $\tau^* \geq 1$, since a history with ``no transmission attempt'' in it is not consistent with $f^2$. Given a tuple $h_{i,\tau^*}=(a_1, a_2,\dots, a_{\tau^*})$, denote by $x_t$ the indicator variable that equals 1 if player $i$ chooses channel 1, and 0 if she chooses channel 2 in round $t \leq \tau^*$. Formally, a protocol as described above is
	\begin{align*}
	g_i = g_{i}(h_{i,\tau^*}) \triangleq 
	\begin{cases}  
	& \left(\text{Pr}\{X_{i,t} = 1 \} = x_t , \quad \text{Pr}\{X_{i,t} = 2 \} = 1-x_t \right) \quad \text{, for } 1 \leq t \leq \tau^* \\
	& f_{i,t}^2 \quad \text{, for } t > \tau^* ,  
	\end{cases}
	\end{align*}
	This process where a single player $i$ uses some protocol $g_i$ and has a latency according to $g_i$ and the other players' fixed protocols, can be modelled as a Partially Observable Markov Decision Process (POMDP) with infinite states; in this POMDP, each state is determined by the transmission history of player $i$ and the number of pending players including $i$, with an additional absorbing state where $i$ goes after successfully transmitting; player $i$'s belief state at any time $t$ is determined by her belief state at time $t-1$, the action she chose at time $t-1$, and her observation (e.g. her transmission history up to $t-1$). 
	
	The fact that we consider acknowledgement-based protocols together with the fact that the partial protocol profile $f_{-i}^2$ which produces our POMDP consists of memoryless and time-independent protocols, make the states of our POMDP be independent of player $i$'s history. We now remark that, regardless of the action taken in some belief state from player $i$ playing $g_i$, the transition probabilities between belief states are independent of time. In particular, denote by $\langle m, t \rangle$ a state with $m$ pending players including player $i$ at time $t \geq 1$, and by $\langle \times \rangle$ the unique absorption state where $i$ finds herself after successful transmission. 
	We write $p_{x}^{y}$ to denote the transition probability to go from state $\langle x \rangle$ to state $\langle y \rangle$. It is easy to see that the transition probabilities among belief states with $1 \leq t \leq \tau^*$ are
	
	\begin{equation*}
	\begin{rcases*}
	p_{m,t}^{\times} = \left(\frac{1}{2}\right)^{m-1} \\
	p_{m,t}^{m-1, t+1} = (m-1)\left(\frac{1}{2}\right)^{m-1} \\
	p_{m,t}^{m, t+1} = 1 - m \left(\frac{1}{2}\right)^{m-1}
	\end{rcases*}
	\forall 3 \leq m \leq n, \quad 1 \leq t \leq \tau^* ,
	\end{equation*}
	\begin{align*}
	\text{and} \quad p_{2,t}^{\times} = \frac{1}{2}, \quad p_{2,t}^{2,t+1} = \frac{1}{2}.
	\end{align*}
	
	%
	
	Observe that the above transition probabilities of any state for which $1 \leq t \leq \tau^*$ are identical to those of equations (\ref{eq (n,2) 1}) and (\ref{eq (n,2) 2}) in the proof of Lemma \ref{exp latency (n,2)}; obviously for $t > \tau^*$ the same holds because player $i$ has switched back to protocol $f^2$.
	Since player $i$'s actions do not affect the transition probabilities of the resulting belief states, the above POMDP reduces to a Markov chain that is in fact identical to the one defined in the proof of Lemma \ref{exp latency (n,2)}, thus $D(g_{i},n) = D(f_{i}^2,n) = 2^n/n$. 
	
	The natural explanation for our POMDP resulting to the above Markov chain is that, if for a given round all players have a  given probability of transmission (not necessarily 1) uniformly distributed on the channels and a single deviator picks an arbitrary distribution on the channels for the same probability of transmission (in this case 1), then: (a) the probability with which she transmits successfully remains unchanged because each channel is blocked with equal probability $(1 - 1/2^{n-1})$ by the rest of the players, and (b) the probabilities with which a specific number $s$ of players (excluding $i$) transmit successfully remain unchanged because, the probability of $s$ players successfully transmitting conditional on $i$ choosing any of the channels is the same (due to the uniform distributions on the channels by the rest of the players) regardless of the channel chosen by $i$. 
	
	\underline{Remark:} The above arguments hold also in the case of any number $k \geq 1$ of channels when an anonymous, memoryless protocol $f$ is used by all players except $i$, where $f$ is defined by a probability $0 < z \leq 1$ that is split uniformly on the channels in every time-step (in our proof, $k=2$ and $z=1$ for all $t > 0$). In such a case the POMDP is reduced to a corresponding Markov chain that is produced when all players follow $f$.
\end{proof}

\begin{lemma}\label{cond ii-b}
	For $2 \leq n \leq 4$ players and $k=2$ channels, any player $i$ that follows protocol $r_{i} \notin \mathcal{G}^{f^2}$ in the profile $(f_{-i}^2,r_{i})$, where $f^2$ is defined in (\ref{protocol(n,2)}), has expected latency at least $2^n / n$.
\end{lemma}

\begin{proof}
	Consider the contention game with fixed number of players $n \in \{2,3,4\}$ and 2 channels. $n-1$ players use protocol $f^2$ and a player $i \in [n]$ uses some protocol $r_{i}=r_{i}(h_{i,\tau^*}) \notin \mathcal{G}^{f^2}$ as defined in (\ref{prot: g}), for some $\tau^* \geq 1$. It is sufficient to show that the lemma holds, when $r_{i}$ is a best response to $f_{-i}^2$, where $r_{i}$ is constrained to be inconsistent with $(f_{-i}^2,f_i^2)$. Therefore, among such best responses $r_{i}$ there has to be one with a round $t < \infty$ for which Pr$\{X_{i,t} = 0\} > 0$ by definition of inconsistency. Let us focus on the smallest such $t$ which we will call from now on $t_0$, i.e. $t_0 \triangleq \inf\{t : \text{Pr}\{X_{i,t} = 0\} > 0\}$. Let us now define the set of protocols $r_{i}(h_{i,t_0}) \notin \mathcal{G}^{f^2}$ for the aforementioned $t_0$. There are two categories of such protocols: Category (1) has $a_{t_0} \neq 0$, and Category (2) has $a_{t_0} = 0$. Each of those categories is partitioned in two other categories: Category (I) has Pr$\{X_{i,t} = 0\} = 0$ for every $t > t_0$, and Category (II) has Pr$\{X_{i,t} = 0\} > 0$ for some $t > t_0$. The categories are presented in Table \ref{table: categories} below.
	
	\begin{table}[h!]
		\centering
		\begin{tabular}{|l|l|}
			\hline
			\textbf{Category 1} & $a_{t_0} \neq 0$ \\ \hline
			\textbf{Category 2} & $a_{t_0} = 0$ \\ \hline
		\end{tabular}
		\text{ } \text{ } \text{ }
		\begin{tabular}{|l|l|}
			\hline
			\textbf{Category I} & $ \forall t > t_0$: Pr$\{X_{i,t} = 0\} = 0$  \\ \hline
			\textbf{Category II} & $\exists t > t_0$: Pr$\{X_{i,t} = 0\} > 0$  \\ \hline
		\end{tabular}
		~\\
		~\\
		~\\
		\caption{The categories of protocol $r_{i}(h_{i,t_0})$.}\label{table: categories}
	\end{table}
	
	Right before time $t_0$ there are $n$ possible cases that could have occurred: $m$ players are pending including player $i$, for $1 \leq m \leq n$. In each of those cases we want to find the expected future latency of a player $i$ that unilaterally uses protocol $r_{i}(h_{i,t_0})$, given history $\vv{h}_{i,t_0 - 1}$, and given that the pending players right before time $t_0$ are $m$; we will denote this by $F_{i,m}^{(\vv{f^2}_{-i},r_{i})}(\vv{h}_{i,t_0 - 1})$. We will prove our claim step by step, starting from protocols of Category (I) which are easier to analyze, and move on to protocols of Category (II); we start the analysis from the case with the least possible players and build up to the required number of players.
	
	Starting with Category (1-I), the analysis of the proof of Lemma \ref{cond ii-a} implies that these protocols $r_{i}$ yield the same expected latency as $f_{i}^2$ in the tuple $f^2$, since their process' Markov chain is identical to this of the case $(f_{-i}^2,f_i^2)$. For Category (2-I), player $i$ does not transmit at $t_0$. Given that right before $t_0$ there are $m$ pending players including $i$, at $t_0$ either all $m$ players remain pending, or $m-1$, or $m-2$; the first event occurs when none of the $m-1$ players using protocol $f^2$ at $t_0$ transmitted successfully, the second when only one of them did, and the third when two of them did. The probability for each of those events is $P_{m-1}(x)$, where $x$ is the number of players that transmit successfully, and can be found in (\ref{prob P_n(x)}) for $k=2$ and $z=1$. Therefore we have,
	\begin{align}\label{equation: 7}
	F_{i,m}^{(\vv{f^2}_{-i},r_{i})}(\vv{h}_{i,t_0 - 1}) = 1 &+ P_{m-1}(0)F_{i,m}^{(\vv{f^2}_{-i},f_{i}')}(\vv{h}_{i,t_0}) + P_{m-1}(1)F_{i,m-1}^{(\vv{f^2}_{-i},f_{i}')}(\vv{h}_{i,t_0})  \nonumber \\
	&+ P_{m-1}(2)F_{i,m-2}^{(\vv{f^2}_{-i},f_{i}')}(\vv{h}_{i,t_0}) \nonumber \\
	= 1 &+ P_{m-1}(0)D(f_{i}^2,m) + P_{m-1}(1)D(f_{i}^2,m-1) + P_{m-1}(2)D(f_{i}^2,m-2),
	\end{align}
	where $f_{i}'$ is the protocol followed by $i$ for $t>t_0$.
	For $m=1$ it is $P_{0}(0) = 1$, and $P_{0}(1) = P_{0}(2) = 0$. For $m = 2$ it is $P_{1}(0) = P_{1}(2) = 0$, and $P_{1}(1) = 1$. For $m=3$ it is $P_{2}(0) = P_{2}(2) = \frac{1}{2}$, and $P_{2}(1) = 0$. For $m = 4$ it is $P_{3}(0) = 1 - 3\left(\frac{1}{2}\right)^{2}$, $P_{3}(1) = 3\left(\frac{1}{2}\right)^{2}$, and $P_{3}(2) = 0$.
	
	Now, using (\ref{equation: 7}), we can see that for $1 \leq m \leq 4$ it is $F_{i,m}^{(\vv{f^2}_{-i},r_{i})}(\vv{h}_{i,t_0 - 1}) \geq F_{i,m}^{(\vv{f^2}_{-i},f_{i})}(\vv{h}_{i,t_0 - 1}) = D(f_{i}^2,m) = 2^m / m$. In particular,
	\begin{align*}
	&\text{for } m=1: \quad 2 \geq 1 , \\
	&\text{for } m=2: \quad 2  \geq 2 , \\
	&\text{for } m=3: \quad \frac{17}{6}  \geq \frac{8}{3} , \text{ and} \\
	&\text{for } m=4: \quad 4  \geq 4 .
	\end{align*}
	Equivalently, $C_{i,m}^{(\vv{f^2}_{-i},r_{i})}(\vv{h}_{i,t_0 - 1}) \geq C_{i,m}^{(\vv{f^2}_{-i},f_{i})}(\vv{h}_{i,t_0 - 1})$, and therefore, due to (\ref{expected 1}),  $C_{i,m}^{(\vv{f^2}_{-i},r_{i})}(\vv{h}_{i,0}) \geq C_{i,m}^{(\vv{f^2}_{-i},f_{i})}(\vv{h}_{i,0})$ for $1 \leq m \leq 4$. Thus, for Category (I) and $2 \leq n \leq 4$, condition (ii-b) of Lemma \ref{eq char 2} holds.
	
	For Category (1-II), we prove our claim for $1 \leq m \leq 4$ pending players right before $t_0$. For $m=1$, obviously $F_{i,1}^{(\vv{f^2}_{-i},r_{i})}(\vv{h}_{i,t_0 - 1}) = 1$, which is also the minimum possible when only one player is pending. For $m=2$, we have
	\begin{align*}
	F_{i,2}^{(\vv{f^2}_{-i},r_{i})}(\vv{h}_{i,t_0 - 1}) &= 1 + \text{Pr}\{\text{No player transmits successfully}\} F_{i,2}^{(\vv{f^2}_{-i},f_{i}')}(\vv{h}_{i,t_0})
	\end{align*}
	Now, given that the protocol $f^2$ used by all players apart from $i$ is time-independent, it should be $F_{i,2}^{(\vv{f^2}_{-i},r_{i})}(\vv{h}_{i,t_0 - 1}) = F_{i,2}^{(\vv{f^2}_{-i},f_{i}')}(\vv{h}_{i,t_0})$. Because if \newline $F_{i,2}^{(\vv{f^2}_{-i},r_{i})}(\vv{h}_{i,t_0 - 1}) < F_{i,2}^{(\vv{f^2}_{-i},f_{i}')}(\vv{h}_{i,t_0})$ or $F_{i,2}^{(\vv{f^2}_{-i},r_{i})}(\vv{h}_{i,t_0 - 1}) > F_{i,2}^{(\vv{f^2}_{-i},f_{i}')}(\vv{h}_{i,t_0})$, then $f_{i}'$ is not a best response; in the former situation player $i$ would prefer $r_{i}(h_{i,t_0})$ over $f_{i}'$; in the latter situation she would prefer a modified protocol $r_{i}(h_{i,t_0}')$ with Pr$\{X_{i,t_0} \neq 0\} = 0$ over the current $r_{i}(h_{i,t_0})$, respectively. The probability of no player transmitting successfully in $t_0$ is $1/2$, thus we get $F_{i,2}^{(\vv{f^2}_{-i},r_{i})}(\vv{h}_{i,t_0 - 1}) = 2 = F_{i,2}^{(\vv{f^2}_{-i},f_{i}^2)}(\vv{h}_{i,t_0 - 1})$, which implies $C_{i,2}^{(\vv{f^2}_{-i},r_{i})}(\vv{h}_{i,t_0 - 1}) = C_{i,2}^{(\vv{f^2}_{-i},f_{i}^2)}(\vv{h}_{i,t_0 - 1})$.
	
	For $m=3$, we have
	\begin{align}\label{equation: 4}
	F_{i,3}^{(\vv{f^2}_{-i},r_{i})}(\vv{h}_{i,t_0 - 1}) = 1 &+ \text{Pr}\{\text{No player transmits successfully}\} F_{i,3}^{(\vv{f^2}_{-i},f_{i}')}(\vv{h}_{i,t_0}) \nonumber \\
	&+ \text{Pr}\{\text{Exactly 1 player other than $i$ transmits successfully}\} F_{i,2}^{(\vv{f^2}_{-i},f_{i}')}(\vv{h}_{i,t_0})
	\end{align}
	From the previous step, we know that a best response to $f_{-i}$ when there are $2$ players pending including $i$ yields expected latency to $i$ equal to 2. Also, the probability that exactly one player other than $i$ transmits successfully when there are 3 players pending, is 1/2. So, (\ref{equation: 4}) gives
	\begin{align*}
	F_{i,3}^{(\vv{f^2}_{-i},r_{i})}(\vv{h}_{i,t_0 - 1}) \geq 2 &+ \text{Pr}\{\text{No player transmits successfully}\} F_{i,3}^{(\vv{f^2}_{-i},f_{i}')}(\vv{h}_{i,t_0})
	\end{align*}
	
	Again, given that the protocol $f$ used by all players apart from $i$ is time-independent, it should be $F_{i,3}^{(\vv{f^2}_{-i},r_{i})}(\vv{h}_{i,t_0 - 1}) = F_{i,3}^{(\vv{f^2}_{-i},f_{i}')}(\vv{h}_{i,t_0})$ for the same reasons explained in the case of $m=2$. The probability of no player transmitting successfully in $t_0$ is $1/2$, thus we get $F_{i,3}^{(\vv{f^2}_{-i},r_{i})}(\vv{h}_{i,t_0 - 1}) \geq 8/3 = F_{i,3}^{(\vv{f^2}_{-i},f_{i}^2)}(\vv{h}_{i,t_0 - 1})$, which implies $C_{i,3}^{(\vv{f^2}_{-i},r_{i})}(\vv{h}_{i,t_0 - 1}) \geq C_{i,3}^{(\vv{f^2}_{-i},f_{i}^2)}(\vv{h}_{i,t_0 - 1})$.
	
	Finally, for $m=4$, we have
	\begin{align}\label{equation: 5}
	F_{i,4}^{(\vv{f^2}_{-i},r_{i})}(\vv{h}_{i,t_0 - 1}) = 1 &+ \text{Pr}\{\text{No player transmits successfully}\} F_{i,4}^{(\vv{f^2}_{-i},f_{i}')}(\vv{h}_{i,t_0}) \nonumber \\
	&+ \text{Pr}\{\text{Exactly 1 player other than $i$ transmits successfully}\} F_{i,3}^{(\vv{f^2}_{-i},f_{i}')}(\vv{h}_{i,t_0})
	\end{align}
	From the previous step, we know that a best response to $f_{-i}^2$ when there are $3$ players pending including $i$ yields expected latency to $i$ at least $8/3$. Also, the probability that exactly one player other than $i$ transmits successfully when there are 4 players pending, is 3/8. So, (\ref{equation: 5}) gives
	\begin{align*}
	F_{i,4}^{(\vv{f^2}_{-i},r_{i})}(\vv{h}_{i,t_0 - 1}) \geq 2 &+ \text{Pr}\{\text{No player transmits successfully}\} F_{i,4}^{(\vv{f^2}_{-i},f_{i}')}(\vv{h}_{i,t_0})
	\end{align*}
	
	Again, given that the protocol $f^2$ used by all players apart from $i$ is time-independent, it should be $F_{i,4}^{(\vv{f^2}_{-i},r_{i})}(\vv{h}_{i,t_0 - 1}) = F_{i,4}^{(\vv{f^2}_{-i},f_{i}')}(\vv{h}_{i,t_0})$ for the same reasons explained for $m \in \{2,3\}$. The probability of no player transmitting successfully in $t_0$ is $1/2$, thus we get $F_{i,4}^{(\vv{f^2}_{-i},r_{i})}(\vv{h}_{i,t_0 - 1}) \geq 4 = F_{i,4}^{(\vv{f^2}_{-i},f_{i}^2)}(\vv{h}_{i,t_0 - 1})$, which implies $C_{i,4}^{(\vv{f^2}_{-i},r_{i})}(\vv{h}_{i,t_0 - 1}) \geq C_{i,4}^{(\vv{f^2}_{-i},f_{i}^2)}(\vv{h}_{i,t_0 - 1})$.
	Thus, for Category (1-II) and $2 \leq n \leq 4$, condition (ii-b) of Lemma \ref{eq char 2} holds.
	
	Now we proceed with the proof of the statement for $1 \leq m \leq 4$ for the final category, namely Category (2-II), using the results from Category (1-II). For every $m \geq 1$, equation (\ref{equation: 6}) holds. For $m=1$, we have
	\begin{align*}
	F_{i,1}^{(\vv{f^2}_{-i},r_{i})}(\vv{h}_{i,t_0 - 1}) = 1 + 1 \cdot F_{i,1}^{(\vv{f^2}_{-i},f_{i}')}(\vv{h}_{i,t_0}) \geq 2 ,
	\end{align*}
	where the above inequality comes from the fact that the minimum expected future latency for $m=1$ is 1 (found in Category (1-II)). By applying the same methodology for $2 \leq m \leq 4$ we have 
	\begin{align*}
	F_{i,m}^{(\vv{f^2}_{-i},r_{i})}(\vv{h}_{i,t_0 - 1}) \geq 3 + \left[1 - (m-1)\left(\frac{1}{2}\right)^{m-2}\right]\frac{2^m}{m} \geq \frac{2^m}{m} = F_{i,m}^{(\vv{f^2}_{-i},f_{i}^2)}(\vv{h}_{i,t_0 - 1}).
	\end{align*}
	
	Then, by taking into account our lower bounds for $F_{i}^{(\vv{f^2}_{-i},r_{i})}(\vv{h}_{i,t})$ when $0 \leq t \leq t_0 - 1$ and for all possible numbers $m$ of remaining players (including $i$), we get
	\begin{align*}
	F_{i}^{(\vv{f^2}_{-i},r_{i})}(\vv{h}_{i,t_0 - 1}) \geq F_{i}^{(\vv{f^2}_{-i},f_{i}^2)}(\vv{h}_{i,t_0 - 1}), \text{ which implies } C_{i}^{(\vv{f^2}_{-i},r_{i})}(\vv{h}_{i,t_0 - 1}) \geq C_{i}^{(\vv{f^2}_{-i},f_{i}^2)}(\vv{h}_{i,t_0 - 1}).
	\end{align*}
	Then, from Corollary \ref{best response} and equation (\ref{expected 1}) it is $C_{i}^{(\vv{f^2}_{-i},f_{i}')}(\vv{h}_{i,0}) \geq C_{i}^{(\vv{f^2}_{-i},f_{i}^2)}(\vv{h}_{i,0})$	and this completes the proof.
\end{proof}


\begin{theorem}\label{thm: eq (n,2)}
	For $n \in \{2,3,4\}$ players and $k = 2$ channels, $f^2$ is an equilibrium protocol with expected latencies $2$, $8/3$ and $4$, respectively.
\end{theorem}

\begin{proof}
	By combining Lemma \ref{cond ii-a}, Lemma \ref{cond ii-b} and the equilibrium characterization of Lemma \ref{eq char 2}.
\end{proof}

\paragraph{\textbf{n players - 3 transmission channels.}}
Here, by employing our characterizations from Subsection \ref{ack-based characterizations}, we give an acknowledgement-based, equilibrium protocol for $n \in \{2,3,4,5\}$ players and $k=3$ channels.  

\begin{theorem}
	For $n \in \{2,3,4,5\}$ players and $k = 3$ channels, $f^3$ defined in (\ref{protocol(n,2)}) is an equilibrium protocol with expected latencies $3/2$, $15/8$, $189/80$ and $597/200$, respectively.
\end{theorem}

We omit the proof because the proof idea is the same as that of Theorem \ref{thm: eq (n,2)}. However, the analysis here is done for each value of $n$ separately, since we do not have a closed form (similar to that of Lemma \ref{exp latency (n,2)}) for the expected latency of $n$ players using protocol $f^3$ for 3 channels. This is because, although using standard Markov chain techniques a linear recurrence relation of the expected latency is easily found, this recurrence relation has non-constant coefficients, for which - to our knowledge - there are no techniques in the literature to solve them\footnote{We note that reducing the recurrence relation to one with constant coefficients using already existing techniques did not work.}.

\section{Equilibria for Ternary Feedback Protocols}\label{section: ter}

In this section we consider anonymous protocols with \textit{ternary feedback}, that is, a pending player knows at every time $t$ the number $m \leq n$ of pending players. This knowledge is given to each player regardless of her transmission history.  

\subsection{Nash equilibrium characterization}\label{ter: history-dep char}

Here we give a characterization of FIN-EQ protocols for $n \geq 1$ players and $k=2$ channels in the general history-dependent case for ternary feedback.

\begin{theorem}
	There exists an anonymous, history-dependent, equilibrium protocol with ternary feedback for $n$ players and 2 transmission channels.
\end{theorem}

\begin{proof}
	Suppose $n \geq 2$ players use the same protocol $f$ in a system with 2 available transmission channels. At time $t$, the decision rule of a player with history $h_{t-1}$ among $m$ pending players is described by the probabilities with which she will transmit on channel 1 and channel 2, i.e. $p_{m,t}^{i,1}$ and $p_{m,t}^{i,2}$ respectively. These transmission probabilities, in general, depend on the history $h_t$ of the respective player, however $t$ is used instead as a subscript in order to make the notation lighter. 
	
	Also, suppose that the anonymous protocol $f$ is an equilibrium and also that $p_{m,t}^{i,1} \neq p_{m,t}^{i,2}$. Without loss of generality $p_{m,t}^{i,1} > p_{m,t}^{i,2}$. Then a player could unilaterally deviate at round $t$ and choose to transmit on channel 2 with probability 1, thus maximizing her probability of success. Therefore, in an anonymous, equilibrium protocol, for every history $h_{t-1}$ and every number $m \geq 2$ of pending players, each player assigns equal transmission probabilities to the channels. Hence we drop the channel indicator superscript - along with the player indicator superscript -  and write $p_{m,t}$. Note that $p_{m,t} \in [0,\frac{1}{2}]$.
	
	We will slightly abuse the notation here and write $C_{m}(h_{t})$ and $F_{m}(h_{t})$ for the \textit{expected cost} of a player (e.g. Alice) and the \textit{expected future cost} of a player respectively, at time $t \geq 0$, given history $h_t$, where there are $1 \leq m \leq n$ pending players. Note that, since the protocol is symmetric, we have replaced the subscript that indicates the player's identity with the one that indicates the number of pending players, and we also have omitted the superscript $f$.
	
	We have\footnote{The probabilities are correct by defining $0^0 = 1$.}
	\begin{align*}
	C_{m}(h_{t}) = P_{m}^{\times} \cdot (t+1) + P_{m}^{m-1} \cdot C_{m-1}(h_{t+1}) + P_{m}^{m-2} \cdot C_{m-2}(h_{t+1}) + P_{m}^{m} \cdot C_{m}(h_{t+1}) \\
	\text{or equivalently,} \quad F_{m}(h_{t}) = 1 + P_{m}^{m-1} \cdot F_{m-1}(h_{t+1}) + P_{m}^{m-2} \cdot F_{m-2}(h_{t+1}) + P_{m}^{m} \cdot F_{m}(h_{t+1}) ,
	\end{align*}	
	\begin{align*}
	\text{where for $m \geq 2$:} \quad P_{m}^{\times} =& \text{Pr}\{\text{Alice transmits successfully} \} \\
	=& 2p_{m,t}(1-p_{m,t})^{m-1} , \\ 
	P_{m}^{m-1} =& \text{Pr}\{\text{Exactly 1 player other than Alice transmits successfully} \} \\
	=& 2(m-1)p_{m,t}\left[ (1-p_{m,t})^{m-1} - (m-1)p_{m,t}(1 - 2p_{m,t})^{m-2} \right] , \\
	P_{m}^{m-2} =& \text{Pr}\{\text{Exactly 2 players other than Alice transmit successfully} \} \\
	=& (m-1)(m-2)p_{m,t}^2 (1 - 2p_{m,t})^{m-2} , \\
	P_{m}^{m} =& \text{Pr}\{\text{No player transmits successfully} \} = 1 - P_{m}^{\times} - P_{m}^{m-1} - P_{m}^{m-2} .
	\end{align*}
	For $m=1$ the pending player has probability of no transmission equal to zero, therefore $F_{1}(h_{t}) = 1$ for every history $h_t$.
	
	Now, given that $m \geq 2$ players are pending, the equilibrium protocol cannot assign to them probability $p_{m,t}=0$ at any time $t$. That is because a unilateral deviator that surely transmitted to a channel would be successful and therefore she would acquire strictly smaller latency than any other player. Since transmission to both channels is in the support of the decision rule of a player at time $t$, both sure transmission attempt to some channel and no transmission should yield the same expected latency to a player. In the sequel we will use the expected future latency $F_{m}(h_t)$ for our analysis. The expected future latency of Alice when she surely transmits on an arbitrary channel in round $t$ with $m \geq 2$ pending players (including herself) is
	\begin{align}\label{F:transmission=1}
	F_{m}(h_{t}) = 1 + Q_{m}^{m-1} \cdot F_{m-1}(h_{t+1}) + (1 - Q_{m}^{\times} - Q_{m}^{m-1}) \cdot F_{m}(h_{t+1}) ,
	\end{align}	
	\begin{align} \label{Q def}
	\text{where for $m \geq 3$:} \quad Q_{m}^{\times} =& \text{Pr}\{\text{Alice transmits successfully} \}  \nonumber \\
	=& (1-p_{m,t})^{m-1} , \nonumber \\ 
	Q_{m}^{m-1} =& \text{Pr}\{\text{Exactly 1 player other than Alice transmits successfully} \} \nonumber  \\
	=& (m-1)p_{m,t}\left[ (1-p_{m,t})^{m-2} - (1 - 2p_{m,t})^{m-2} \right] ,  \nonumber  \\
	Q_{m}^{m} =& \text{Pr}\{\text{No player transmits successfully} \} = 1 - Q_{m}^{\times} - Q_{m}^{m-1} , \nonumber \\
	\text{for $m=2$:} \quad Q_{2}^{\times} = & 1 - p_{m,t}, \quad  Q_{2}^{1} = 0, \quad  Q_{2}^{2} = p_{m,t} .
	\end{align}
	
	The expected future latency of Alice when she surely does not attempt transmission in round $t$ with $m \geq 2$ pending players (including herself) is
	\begin{align}\label{F:transmission=0}
	F_{m}(h_{t}) = 1 + S_{m}^{m-1} \cdot F_{m-1}(h_{t+1}) + S_{m}^{m-2} \cdot F_{m-2}(h_{t+1}) + (1 - S_{m}^{m-1} - S_{m}^{m-2}) \cdot F_{m}(h_{t+1}) ,
	\end{align}	
	\begin{align} \label{S def}
	\text{where for $m \geq 3$:} \quad S_{m}^{m-1} =& \text{Pr}\{\text{Exactly 1 player other than Alice transmits successfully} \}  \nonumber \\
	=& 2(m-1)p_{m,t}\left[ (1-p_{m,t})^{m-2} - (m-2)p_{m,t}(1 - 2p_{m,t})^{m-3} \right] , \nonumber  \\
	S_{m}^{m-2} =& \text{Pr}\{\text{Exactly 2 players other than Alice transmit successfully} \}  \nonumber \\
	=& (m-1)(m-2)p_{m,t}^2 (1 - 2p_{m,t})^{m-3} ,  \nonumber \\
	S_{m}^{m} =& \text{Pr}\{\text{No player transmits successfully} \} = 1 - S_{m}^{m-1} - S_{m}^{m-2} , \nonumber \\
	\text{for $m=2$:} \quad S_{2}^{1} = & 2p_{m,t}, \quad  S_{2}^{0} = 0, \quad  S_{2}^{2} = 1 - 2p_{m,t} .
	\end{align}
	
	By equating the right-hand sides of (\ref{F:transmission=1}) and (\ref{F:transmission=0}) we get the probability $p_{m,t}$ as a function of expected future costs $F_{m-1}(h_{t+1}), F_{m-2}(h_{t+1})$ and $F_{m}(h_{t+1})$. 
\end{proof}

The equilibrium probability that depends on the number of pending players $m$ and defines the equilibrium protocol, although guaranteed to exist when expected (future) latencies are finite, is difficult to be expressed in closed form. Contrary to the case of a single channel studied in \cite{FMN07}, where $p_{m,t}$ can be nicely expressed as a function of $F_{m-1}(h_{t+1})$ and $F_{m}(h_{t+1})$ in closed form, this does not seem to be the case in the current setting.

We should mention here that in the single-channel setting studied in \cite{FMN07} the decision rule $p_{m,t}=1$ for $m \geq 3$ is in equilibrium. However, in the case of two channels, a similar result (e.g. $p_{m,t}=1/2$) for any number of pending players does not seem to hold. Indeed, in time $t$ with $m \geq 5$ pending players playing $p_{m,t}=1/2$, the best response with strictly better expected latency is $p_{m,t}=0$.

\subsection{History-independent FIN-EQ protocols}\label{ter: history-indep char}

Let us now consider anonymous, history-independent protocols, that is, protocols whose decision rules depend only on the number $1 \leq m \leq n$ of pending players. Now, the decision rule $p_m$ of the players does not depend on their transmission history (and therefore on time as well), hence a player's expected future latency $F_m$ does not depend on her transmission history. In this class of protocols the following theorem fully characterizes the equilibria.

\begin{theorem}\label{ter: hist-ind}
	There exists a unique, anonymous, history-independent, equilibrium protocol with ternary feedback for $n$ players and $2$ transmission channels, which is: any player among $2 \leq m \leq n$ remaining players, for every $t \geq 1$ attempts transmission to each channel with equal probability $p_m$. This probability is $\Theta(\frac{1}{\sqrt{m}})$ and yields expected future latency $e^{\Theta(\sqrt{m})}$ for every player.
\end{theorem}

\begin{proof}
	By manipulating the equilibrium conditions (\ref{F:transmission=1}) and (\ref{F:transmission=0}) we find 
	
	\begin{align}\label{F_m}
	F_m = \frac{\left[ Q_{m-1}^{m-2} S_{m}^{m-1} + S_{m}^{m-2} (1 - Q_{m-1}^{m-1}) \right] - Q_{m}^{m-1} (Q_{m-1}^{m-2} - S_{m}^{m-2})}{(1 - Q_{m}^{m}) \left[ Q_{m-1}^{m-2} S_{m}^{m-1} + S_{m}^{m-2} (1 - Q_{m-1}^{m-1}) \right] - Q_{m-1}^{m-2} Q_{m}^{m-1} (1 - S_{m}^{m})} . 
	\end{align}
	
	From this we can also get $F_{m-1}$, thus, replacing these two in relation (\ref{F:transmission=1}), which, in the history-independent case becomes 
	\begin{align}\label{recur 1}
	(1 - Q_{m}^{m}) F_m = 1 + Q_{m}^{m-1} F_{m-1} ,
	\end{align}
	we get the recurrence relation for the transmission probability $p_{m}$ to each channel. The resulting recurrence relation of $p_m$ is non-linear with non-constant coefficients and for its form there is no methodology in the literature that solves it - to the authors' knowledge. However, we can find the asymptotic behaviour of $p_{m}$ in the following way. 
	
	First, we show by induction that $p_m$ is uniquely determined. The recurrence relation of $p_m$ holds for $m \geq 2$ since our probabilities $Q$ and $S$ are defined for this domain only. That is because probabilities $Q$ and $S$ stem from the requirement that ``transmission" and ``no transmission" are both in the support of the decision rule for a player, which is not true in the case of $m=1$. As a base case of our induction we use $m=2$, for which we find from (\ref{F:transmission=1}) and (\ref{F:transmission=0}) as unique solution the pair ($p_2 = 1/2$, $F_2 = 2$). Now consider some $m \geq 2$ and assume that all $p_{m'}$ are uniquely determined for every $m'$, $2 \leq m' \leq m$, and thus all $F_{m'}$ are uniquely determined by (\ref{F_m}). Let us replace $m$ with $m+1$ in (\ref{recur 1}), and fix $p_{m}$ and $F_{m}$ with the known ones. This gives us a rational univariate function - let us call it $h$ - of $p_{m+1}$, i.e. $h(p_{m+1}) = (1 - Q_{m+1}^{m+1}) F_{m+1} - 1 - Q_{m+1}^{m} F_{m}$. We would like to find the roots of $h$ in the interval $(0,1/2]$. By substituting $Q_{m+1}^{m+1}$, $Q_{m+1}^{m}$ and $F_m$ from (\ref{Q def}) and (\ref{F_m}) respectively, and then examining the first and second derivative of $h$, we can see that $h(0)=0$, $h$ has its unique minimum for some $p_{m+1}' \in (0,1/2)$, and it is strictly decreasing in $[0, p_{m+1}']$. In $[p_{m+1}',1/2]$ it is strictly increasing and $h(1/2) \geq 0$. Therefore, in $(0,1/2]$ there is a unique root $p_{m+1}^*$ of $h$. 
	
	Now we proceed in showing that the asymptotic behaviour in both sides of the recurrence relation (\ref{recur 1}) is the same for $p_m \in \Theta (1/\sqrt{m})$. First, we express the probabilities $Q$ and $S$ (see sets of equations (\ref{Q def}) and (\ref{S def})) in terms of $Q_{m}^{\times}$, and then we put $p_{m,t} = p_m \in \Theta (1/\sqrt{m})$. This gives: 
	%
	\begin{align*}
	& Q_{m}^{\times} \in e^{- \Theta (\sqrt{m})} , \quad  Q_{m}^{m-1} = Q_{m}^{\times} \cdot f_{1}(m), \quad  Q_{m}^{m-2} = 0, \quad  Q_{m}^{m} = 1 - Q_{m}^{\times} \cdot f_{2}(m) , \quad  \text{and} \\
	& S_{m}^{\times} = 0 , \quad  S_{m}^{m-1} = Q_{m}^{\times} \cdot g_{1}(m) , \quad  S_{m}^{m-2} = \left( Q_{m}^{\times} \right)^2 \cdot g_{2}(m), \quad  S_{m}^{m} = 1 - Q_{m}^{\times} \cdot g_{3}(m) , 
	\end{align*}
	where the functions $f_{1}(m), f_{2}(m), g_{1}(m), g_{3}(m)$ are in $\Theta(\sqrt{m})$ and $g_{2}(m)$ is in $\Theta(m)$. Now that we have described the asymptotic behaviour of the probabilities $Q$ and $S$, we can find the asymptotic behaviour of the expected future latency $F_m$ using (\ref{F_m}). By carefully simplifying the numerator and denominator in the right-hand side of (\ref{F_m}) we get 
	\begin{align*}
	F_m = \frac{1}{ Q_{m}^{\times} \cdot h_{1}(m) } , \quad \text{where } h_{1}(m) \in \Theta(\sqrt{m}).
	\end{align*}
	Recall that $Q_{m}^{\times} \in e^{- \Theta (\sqrt{m})}$, thus $F_m \in e^{\Theta (\sqrt{m})}$. The above formula for $F_m$ also implies that $F_{m-1} = 1/\left( Q_{m-1}^{\times} \cdot h_{2}(m) \right)$, where $h_{2}(m) \in \Theta(\sqrt{m})$. By substituting $F_m$ and $F_{m-1}$ in the recurrence relation (\ref{recur 1}), we show that the asymptotic behaviour in both sides of it are the same, in particular, $\Theta(1)$. This completes the proof.
\end{proof}

The latter result is analogous to the one in \cite{FMN07} that characterizes anonymous, history-independent, equilibrium protocols with ternary feedback for the case of a single channel. However here, the proof methodology is different due to the fact that there is no known technique to express the equilibrium transmission probabilities in closed form, therefore their asymptotic behaviour can only be extracted from a recurrence relation, which, contrary to the one in \cite{FMN07}, is quite complex. Using dynamic programming, we can compute the equilibrium probabilities in linear time; for up to $m=100$ the equilibrium probabilities are presented in Figure \ref{fig-bounds}.

\begin{figure}[h]
	\begin{center}
		\includegraphics[scale=0.52]{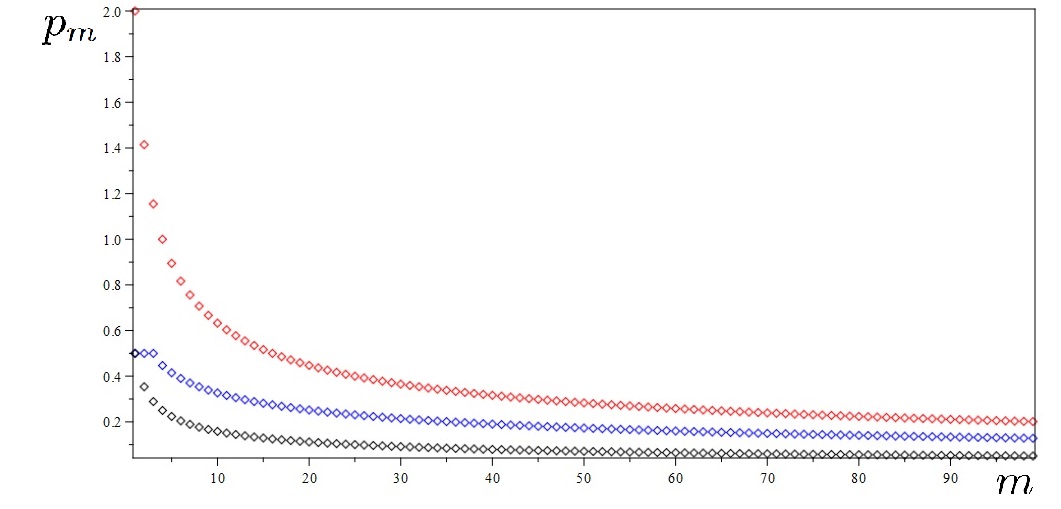}
	\end{center}
	\caption{Blue: the equilibrium probabilities $p_m$ for $2 \leq m \leq 100$. Red: experimental upper bound, function $\frac{2}{\sqrt{m-1}}$. Black: experimental lower bound, function $\frac{1}{2\sqrt{m-1}}$.}\label{fig-bounds}
\end{figure}

\section{IN-EQ Protocols for Both Feedback Classes}\label{eff prot}

Ideally, we would like to find an anonymous, equilibrium protocol that is efficient (i.e. the time until all players transmit successfully is $\Theta(n/k)$ with high probability) and also has finite expected latency. For the case of ternary feedback and a single channel such a protocol was found in \cite{FMN07}. That protocol set a deadline $t_0 \propto n$ after which it prescribed to the players to transmit with probability 1 on the channel at every time. It is easy to see that transmitting surely at every time is an equilibrium for more than 2 players. The idea of that protocol was to employ that ``bad equilibrium'' by putting it after the deadline so as to keep the players that were unsuccessful until $t_0$ for a very long (exponential in $n$) time. This works as a threat to the players, which they try to avoid by adopting a cooperative behaviour; using a history-dependent, equilibrium protocol they attempt transmission with probability low enough so that all of them are successful before the deadline with high probability. After the long part of the protocol, there is a last part that prescribes to the players to use a history-independent, equilibrium protocol (similar to the one we find for the 2-channel case) which has finite expected future latency. Since all three parts of the protocol are in equilibrium, the whole protocol is in equilibrium as well.

However, for the case of multiple channels in both the ternary feedback and acknowledgement-based feedback classes, a protocol with the above structure cannot be constructed as the following theorem shows.
First, let us define the following notion of equilibrium protocol:
By \textit{equilibrium with blocking step (EBS)} we call an anonymous, equilibrium protocol with the property that there exists a time-step ($<\infty$) of the protocol in which every pending player has probability of successful transmission equal to 0.

\begin{theorem}\label{thm: impos}
	In the classes of acknowledgement-based and ternary feedback protocols with $k \geq 2$ channels and $n \geq 2$ players, there exists no equilibrium protocol with blocking step (EBS) and finite expected latency.
\end{theorem}

\begin{proof}
	Assume for the sake of contradiction, that $f$ is an anonymous equilibrium protocol with finite expected latency and it has a blocking step. Suppose all $n$ players follow $f$, therefore the protocol profile is $\vv{f} = (f_1,f_2,\dots, f_n) = (f,f,\dots,f)$. Also, suppose that the blocking step is at $t = t_0$, which means that in any combination of personal histories $\vv{h}_{t_0} = (h_{1, t_0}, h_{2, t_0}, \dots, h_{m, t_0})$ of the $m \leq n$ pending players which happens with positive probability under $\vv{f}$, no player transmits successfully. Additionally, since $f$ is an equilibrium protocol, at $t_0$ the probability that some channel is ``free'' is 0, because if not, a player could deviate unilaterally at $t_0$ by choosing that channel with some positive probability, and thus improve her expected latency.
	
	Consider the set $H$ that contains all combinations of personal histories \newline $\vv{h}_{t_0 -1} = (h_{1, t_0 -1}, h_{2, t_0 -1}, \dots, h_{m, t_0 -1})$. Consider also the subset $H_{>}$ of $H$ that contains all combinations that happen with positive probability, and the subset $H_{0} = H \setminus H_{>}$ containing those that happen with probability 0. 
	For the reasons explained above, any combination $\vv{h}_{t_0 -1}' \in H_{>}$ is characterized by the property that the combination of decision rules of the $m$ pending players at $t_0$ that it produces necessarily has at least 2 players assigning probability 1 on each channel (so that every channel is surely blocked and no player can deviate unilaterally). Any combination $\vv{h}_{t_0 -1}'$ that does not have this property must be in $H_{0}$, otherwise a  player that under $f$ would assign probability 1 on a surely blocked channel at $t_0$, she could unilaterally deviate by assigning at $t_0$ an appropriate positive probability on a channel that is ``free'' with positive probability, and decreasing her expected latency.
	
	Now pick an arbitrary element $\vv{h}_{t_0 -1}'$ of $H_{>}$, and without loss of generality, suppose that players $1,2,\dots, 2k$ block all $k$ channels with probability 1 at $t_0$. That is
	\begin{align}\label{blocking_players}
	f_{1}(h_{1,t_0 -1}) &= \left( \text{Pr}\{X_{1, t_0} = 1 \} = 1, \text{ } \text{Pr}\{X_{1, t_0} \neq 1 \} = 0 \right)  ,  \\
	f_{2}(h_{2,t_0 -1}) &= \left( \text{Pr}\{X_{2, t_0} = 1 \} = 1, \text{ } \text{Pr}\{X_{2, t_0} \neq 1 \} = 0 \right)  , \nonumber \\
	f_{3}(h_{3,t_0 -1}) &= \left( \text{Pr}\{X_{3, t_0} = 2 \} = 1, \text{ } \text{Pr}\{X_{3, t_0} \neq 2 \} = 0 \right)  ,\nonumber  \\
	f_{4}(h_{4,t_0 -1}) &= \left( \text{Pr}\{X_{4, t_0} = 2 \} = 1, \text{ } \text{Pr}\{X_{4, t_0} \neq 2 \} = 0 \right)  , \nonumber \\
		\vdots \nonumber\\
	f_{2k-1}(h_{2k-1,t_0 -1}) &= \left( \text{Pr}\{X_{2k-1, t_0} = k \} = 1, \text{ } \text{Pr}\{X_{2k-1, t_0} \neq k \} = 0 \right)  , \nonumber \\
	f_{2k}(h_{2k,t_0 -1}) &= \left( \text{Pr}\{X_{2k, t_0} = k \} = 1, \text{ } \text{Pr}\{X_{2k, t_0} \neq k \} = 0 \right).  \nonumber
	\end{align}
	
	Consider now that the event $\vv{h}_{t_0 -1}^*$ where all players have the same history $h_{1,t_0 -1}$ right before $t_0$, and observe that happens with positive probability, because the game starts with all players having the exact same history, i.e. the empty history, and can continue having the same history by transmitting to the exact same channels for every $t \leq t_0 -1$ since they will be prescribed identical decision rules by protocol $f$. Therefore $\vv{h}_{t_0 -1}^* \in H_{>}$. In such an event, all players, just like player 1 in (\ref{blocking_players}) above, will transmit with probability 1 on channel 1 at $t_0$. Therefore the remaining $k-1$ channels will be ``free'' at $t_0$, therefore $\vv{h}_{t_0 -1}^* \in H_{0}$ which is a contradiction.
	
	In the above analysis, the arguments about a player being able to unilaterally deviate and decrease her expected latency, need the extra property that the expected latency of the player is finite. Because if the expected latency is infinite, unilateral deviation does not make it finite, therefore she has no incentive to deviate. The proof is complete.
\end{proof}

This impossibility result discourages the search for anonymous, efficient, multiple-channel, equilibrium protocols with the additional property of finite expected latency, since the only candidate that guarantees efficiency seems to be a deadline protocol. Whether no anonymous, efficient, equilibrium protocol with finite expected latency can be found for multiple channels is one of the most interesting open problems that stem from this work.

In the rest of this section we present IN-EQ protocols within the classes of acknowledgement-based and ternary feedback for the general case of $k\geq 1$ channels and any number of $n \geq 2k + 1$ players. For this, we employ the deadline idea introduced in \cite{FMN07} and consequently used in \cite{CGNRS16, CGNRS17}. Our protocols are efficient, even though the expected latency is infinite.

\subsection{Acknowledgemnt-based feedback}\label{ack: eff prot}

%

In the aforementioned companion short paper \cite{CMS18} we provided an efficient deadline protocol with infinite expected latency for $k \geq 1$ channels and $n \geq 2k+1$ players. This protocol generalizes the efficient protocol of \cite{CGNRS16} which deals with a single channel and at least 3 players. The general protocol we present uses their idea, that is, estimating the number of pending players (since it is not known in the acknowledgement-based environment) and adjusting the transmission probabilities of the players accordingly, in order to simulate a socially optimal protocol (see protocol SOP below) that allows all transmission to be successful by time $\Theta(n/k)$ with high probability. Our modification is that, instead of prescribing to the players to always transmit to the single channel once they reach the deadline (so that with some positive probability they get blocked forever), we block all channels with positive probability by prescribing a random assignment of each player to a channel. 

In particular, consider $k \geq 1$ transmission channels, $n \geq 2k + 1$ players, a fixed constant $\beta \in (0,1)$ and a deadline $t_0$ to be determined consequently. The $t_0 - 1$ time steps are partitioned into $r+1$ consecutive intervals $I_1, I_2,\dots, I_{r+1}$ where $r$ is the unique integer in $\left[ -\log_{\beta}n/2 - 1, -\log_{\beta}n/2 \right]$. For any $j \in \{1,2,\dots, r+1\}$ define $n_j = \beta^{j}n/k$. For $j \in \{1,2,\dots, r\}$ the length of interval $I_j$ is $l_j = \lfloor\frac{e}{\beta}n_j\rfloor$. Interval $I_{r+1}$ is special and has length $l_{r+1} = n/k$. We define the following anonymous protocol.

\begin{tcolorbox}
	\textbf{\underline{Protocol $\mathbf{g_1}$}}: \nonumber \\ 
	\text{Every player among $1 \leq m \leq n$ pending players for $t \in I_j$ assigns transmission probability} \\
	\text{$1/\max\{n_j, k\}$ to each channel. Right before the deadline $t_0 = 1 + \sum_{j=1}^{r+1}l_j$ each pending}   \\
	\text{player is assigned to a random channel equiprobably, and for $t \geq t_0$ always attempts}\\
	\text{transmission to that channel.}
\end{tcolorbox}

\begin{lemma}
	Protocol $g_1$ for $n \geq 2k + 1$ players and $k \geq 1$ channels, is an equilibrium protocol and it is also efficient.
\end{lemma}

\begin{proof}
	First we prove that $g_1$ is an equilibrium protocol when $n \geq 2k + 1$. Consider an arbitrary player $i$, and observe that since all players play $g_1$ the probability that all of them will be still pending by $t_0$ is $1/n^{t_0} > 0$. Given that, the probability that player $i$ at $t_0$ will be assigned to the same channel with at least 2 other players is at least the probability that she will be assigned to the same channel with all other players, which is at least $1/k^{n} > 0$. Hence, the probability that player $i$ can find herself in $t=t_0$ pending together with two other players is positive, and in this case she will remain pending forever. Therefore, $i$'s expected latency is $\infty$, and since by any unilateral deviation of $i$ she cannot make the aforementioned event empty, her expected latency will always be $\infty$. Therefore, $g_1$ is an equilibrium protocol.
	
	Now we proceed by showing that $g_1$ is also efficient, that is, all players transmit successfully by time $t_0 = 1 + \sum_{j=1}^{r+1}l_j \in \Theta(n/k)$. The proof of efficiency is essentially the same as that of Theorem 11 in \cite{CGNRS16} and it is omitted. The difference here is that we have tuned $n_j$ and $l_{r+1}$ according to our problem and we have used variable $r$ instead of $k$. As a consequence, this result is the same as that of the aforementioned theorem in \cite{CGNRS16}, except that ours has $n/k$ instead of $n$.
\end{proof}

%
%

\subsection{Ternary feedback}\label{ter: eff prot}

In the ternary feedback setting, the use of the unique history-independent equilibrium (see Subsection \ref{ter: history-indep char}) yields exponential expected latency in the number of players $n$, and additionally, even one player's latency being any polynomial in $n$ happens with exponentially small probability. This fact points to history-dependent protocols as candidates for efficient equilibria. Here, we construct a protocol (Theorem \ref{ter: efficient}) which imposes a heavy cost on any player that does not manage to transmit successfully until a certain deadline-round. This forces any potential deviator to play ``fairly'' until the deadline and follow an anonymous, socially optimal protocol, named $SOP$ (guarantees expected time $\Theta(n/k)$ for all players to pass). 

To prove the main theorem of this subsection we need a series of technical results.
As a first step, we give the general Lemma \ref{balls-bins} that determines the expected number of successful transmissions in a round where $m$ players have a uniform distribution on the channels, and subsequently in Fact \ref{fact max exp} we find the maximum of that expected number. Then, we present another lemma (Lemma \ref{claim:finish-time(n,n)}) that gives an upper bound on the expected finishing time when $m \leq k$. Finally, using all the aforementioned intermediate results, we present a socially optimal protocol in Lemma \ref{opt prot} which is employed in the proof for our main theorem (Theorem \ref{ter: efficient}) that concerns our IN-EQ protocol.

\begin{lemma}\label{balls-bins}
	Consider a single round with $k \geq 1$ channels and $n \geq 1$ players. Assume that for every player the probability of transmission attempt is $z \in [0,1]$ which she splits equally to all $k$ channels. Then, the expected number\footnote{We define $0^0 = 1.$} of players that transmit successfully is $z n \left( 1 - \frac{z}{k} \right)^{n-1}$.
\end{lemma}

\begin{proof}
	For a fixed $z \in [0,1]$, denote by $X_n$ the random variable that indicates how many players transmit successfully in a round with $n$ players. Note that when $z=1$ and $n \geq 2$, the case where $X_n = n-1$ is impossible since in order for some player to have a failed transmission she has to be blocked by someone else.
	
	Our problem reduces to the following balls-and-bins problem: Consider $n$ balls and $k$ bins, where $n \geq 1$. Each ball is thrown with probability $z/k$ to each bin, and not thrown at all with probability $1-z$. Random variable $X_n \in \{0,1,2, \dots, n\}$ now indicates the number of bins that had a single ball after the experiment.
	
	We want to find $\mathbb{E}[X_n]$. For this, we will employ the probability of the event that $x$ bins contain a single ball given that the round started with $n$ balls. Denote by $A_j$ the event that bin $j$ contains a single ball. Also, we define the probabilities of intersections between such events
	\begin{align*}
	p_j = \text{Pr}(A_j), \quad p_{jm} = \text{Pr}(A_j \cap A_m), \quad p_{jml} = \text{Pr}(A_j \cap A_m \cap A_l), \quad \dots
	\end{align*}
	and we write $S_r$ to denote the sums of all distinct $p$'s with $r$ subscripts. That is
	\begin{align*}
	S_1 = \sum_{j=1}^{k} p_j, \quad S_2 = \sum_{j<m} p_{jm}, \quad S_3 = \sum_{j<m<l} p_{jml}, \quad \dots
	\end{align*}
	where the subscripts are in increasing order $j<m<l<\cdots<k$ for uniqueness, so that in the sums each combination appears only once; therefore $S_r$ has $\binom{k}{r}$ terms. In our setting, each term of $S_r$ equals 
	\begin{align*}
	\binom{n}{r}r!\left(\frac{z}{k}\right)^r \left(1 - \frac{rz}{k}\right)^{n-r}
	\end{align*}
	since for specific $r$ bins to contain a single ball there are $\binom{n}{r}$ combinations of $r$ balls, which should occupy the $r$ bins with $r!$ orders. Each of those chosen $r$ balls can fall in a bin with probability $\frac{z}{k}$ and each of the rest $n-r$ balls has to fall in some other than those $r$ bins or not be thrown at all, which happens with probability $1 - \frac{rz}{k}$. So,
	\begin{align*}
	S_r = \binom{k}{r}\binom{n}{r}r!\left(\frac{z}{k}\right)^r \left(1 - \frac{rz}{k}\right)^{n-r}
	\end{align*}
	and by the Inclusion-Exclusion Theorem, the probability that exactly $x$ bins contain a single ball is the following\footnote{For the case where $a < b$ we define $\binom{a}{b} \triangleq 0$ so that the analysis is displayed only once for both cases $n \leq k$ and $n > k$.}
	\begin{align}\label{prob P_n(x)}
	P_{n}(x) &= \sum_{r=x}^{n} (-1)^{r-x} \binom{r}{x} S_r \nonumber \\
	&= \sum_{r=x}^{n} (-1)^{r-x} \binom{r}{x} \binom{k}{r}\binom{n}{r}r!\left(\frac{z}{k}\right)^r \left(1 - \frac{rz}{k}\right)^{n-r}
	\end{align}	
	
	We want to calculate $\mathbb{E}[X_n]$. We have
	\begin{align*}
	\mathbb{E}[X_n] & = \sum_{x=0}^{n} x P_{n}(x)   \\
	& = \sum_{x=0}^{n} \sum_{r=x}^{n} (-1)^{r-x} x \binom{r}{x} \binom{k}{r}\binom{n}{r}r!\left(\frac{z}{k}\right)^r \left(1 - \frac{rz}{k}\right)^{n-r}    \\
	& = \sum_{r=0}^{n} \sum_{x=0}^{r} (-1)^{r-x} x \binom{r}{x} \binom{k}{r}\binom{n}{r}r!\left(\frac{z}{k}\right)^r \left(1 - \frac{rz}{k}\right)^{n-r}    \\
	&  = \sum_{r=0}^{n} \binom{k}{r}\binom{n}{r}r! \left(\frac{z}{k}\right)^r \left(1 - \frac{rz}{k}\right)^{n-r} \sum_{x=0}^{r} (-1)^{r-x} x \binom{r}{x}  \\
	&  = \sum_{r=0}^{n} \binom{k}{r}\binom{n}{r}r! \left(\frac{z}{k}\right)^r \left(1 - \frac{rz}{k}\right)^{n-r} (-1)^r \sum_{x=0}^{r} (-1)^{x} x \binom{r}{x}  \\
	&  = \binom{k}{1}\binom{n}{1} \frac{z}{k} \left(1 - \frac{z}{k}\right)^{n-1} (-1) (-1)  \quad \text{(since } \sum_{x=0}^{r} (-1)^{x} x \binom{r}{x} = -1 \text{ for } r=1, \quad 0 \text{ otherwise)}    \\
	&  = zn \left(1 - \frac{z}{k}\right)^{n-1} .
	\end{align*}
\end{proof}

The following fact shows where the expected number of players of the above theorem is maximized as a function of $z$ (the probability mass devoted to transmission).

\begin{fact}\label{fact max exp}
	Consider the function $f(z) = zn (1- z/k)^{n-1}$, with domain $[0,1]$ and parameters $k \geq 1$, and $n \geq 1$. The maximum of $f$ is attained for $z = \min \{ k/n , 1 \}$.  
\end{fact}

\begin{proof}
	The first and second derivatives of $f$ are
	\begin{align*}
	f'(n) &= n \left( 1 - \frac{zn}{k} \right) \left( 1 - \frac{z}{k} \right)^{n-2} \\
	f''(n) &= n (n-1) \left( 1 - \frac{z}{k} \right)^{n-3} \frac{nz - 2k}{k^2}
	\end{align*}
	When $n < k$, then $f'(z) > 0$ and therefore the global maximum of $f$ is attained for $z = 1$, which gives $f(1) = n(1 - 1/k)^{n-1}$. 
	
	When $n \geq k$, the first derivative of $f$ is 0 for (a) $z=k$ when $n \geq 3$, or (b) $z=k/n$ when $n \geq 1$. Case (a) only works if $k=1$ due to the domain of $z$ and gives $f(1)=0$. $f'(z)$ is positive in $[0,k/n)$, and negative in $(k/n,1)$. Therefore, $f(k/n)=k(1 - 1/n)^{n-1}$ is the global maximum. 
\end{proof}

\begin{lemma}\label{claim:finish-time(n,n)}
	Suppose there are $k \geq 2$ channels and $2 \leq n \leq k$ players and suppose that all players use the following protocol: A player in every time step $t \geq 1$ has a probability of transmission $1/k$ to every channel. Then, the expected time until everyone transmits successfully is upper bounded by $\frac{1}{1-\ln (e-1)} \ln (\frac{n}{2}) + \left( 1 - \frac{1}{k} \right)^{-1}$.
\end{lemma}

\begin{proof}
	Denote by $X_m$ the random variable that indicates how many players transmit successfully in a round $t$ where $m \in \{0,2,\dots,n\}$ players are left. Note that the case where $m=1$ is impossible since in order for some player to have a failed transmission she has to be blocked by someone else. In the next round the expected number of players will be $m - \mathbb{E}[X_m]$. We define the finishing time as the following random variable $T \triangleq inf\{t:m=0\}$ and we would like to find its expectation.
	
	Our problem reduces to the following balls and bins problem: Consider $n$ balls and $k$ bins, where $2 \leq n \leq k$. At time $t=1$ all balls are thrown uniformly at random to the $k$ bins. For all the bins that contain a single ball, these balls are removed, and in the next round $m \in \{0,2,\dots,n\}$ balls remain. At time $t=2$ all $m$ balls are thrown uniformly at random to the $k$ bins. The process continues as long as there are remaining balls. Random variable $X_m \in \{0,1,2, \dots, m\}$ now indicates the number of bins that had a single ball when the respective round started with $m$ balls. Note again that Pr($X_m = m-1$)$=0$ since there is no allocation of balls in the bins such that $m-1$ bins have a single ball. Random variable $T \triangleq inf\{t:m=0\}$ is the finishing time of this process.
	
	We define the function $f(m)$ to be the expected finishing time $\mathbb{E}[T]$ when $m$ players remain. We assume that this function is non-decreasing and concave. Then we have,
	\begin{align}\label{f(m)}
	f(m) &= 1 + \sum_{i=0}^{m}\text{Pr}(X_m=i)f(m-X_m)  \nonumber	\\
	&= 1 + \mathbb{E}[f(m-X_m)]  \nonumber  \\
	&\leq 1 + f\left(\mathbb{E}[m-X_m]\right)	\nonumber \qquad &\text{(concavity of $f$ and Jensen's inequality)} \\
	&= 1 + f(m- \mathbb{E}[X_m])	 \qquad &\text{(linearity of expectation)}
	\end{align}
	
	Now by exploiting the monotonicity of the function $f(m)$ in equation (\ref{f(m)}), and using Lemma \ref{balls-bins} we only need to find a lower bound on $\mathbb{E}[X_m]$. This is easy, since $m \left(1 - \frac{1}{k}\right)^{m-1} \geq m \left(1 - \frac{1}{k}\right)^{k-1} \geq m/e$. Then from equation (\ref{f(m)}) we get
	\begin{align*}
	f(m) & \leq 1 + f\left(m \left( 1 - \frac{1}{e} \right)\right)	\\
	& \leq r + f\left(m \left( 1 - \frac{1}{e} \right)^r \right).
	\end{align*}	
	We use as base case $f(2)$ for which holds that $f(2) = 1 + k \frac{1}{k^2} f(2)$, or equivalently, $f(2) = (1- 1/k)^{-1}$. Then the $r$ for which $m \left( 1 - \frac{1}{e} \right)^r = 2$ finally gives us 
	\begin{align*}
	f(m) \leq \frac{1}{1-\ln (e-1)} \ln (\frac{m}{2}) + \left( 1 - \frac{1}{k} \right)^{-1}
	\end{align*}	
\end{proof}

Let us define the following anonymous, history-independent protocol which we prove to be efficient. However, we remark that it is not in equilibrium, due to Theorem \ref{ter: hist-ind} which characterizes the unique, anonymous, equilibrium protocol that is history-independent.

\begin{tcolorbox}
	\textbf{\underline{Protocol $\mathbf{SOP}$}}: \nonumber \\ 
	\text{Every player among $1 \leq m \leq n$ pending players, in each round $t \geq 1$ assigns transmission} \\
	\text{probability $1/\max\{m, k\}$ to each channel.}  
\end{tcolorbox}

In the sequel, by $e$ we denote the constant named ``Euler's number'', i.e. $e=2.7182\dots$. 

\begin{lemma}\label{opt prot}
	Protocol $SOP$ for $k \geq 1$ channels and $n > k$ players has expected finishing time $O((n-k)/k)$.
\end{lemma}

\begin{proof}
	Suppose protocol $SOP$ as stated in the theorem is used. Then, the transmission probability of each player in each round is uniform on the set of channels $K$. Using the framework of Lemma \ref{balls-bins}, according to protocol $SOP$ for variable $z$ we have $z=\min\{ k/m, 1 \}$, and we know from Fact \ref{fact max exp} that this value maximizes the number of successful transmissions in a round with $m$ players. Denote by $X_m$ the random variable that keeps track of the number of successful transmissions in a single round with $m > k$ pending players. Then, according to Lemma \ref{balls-bins}, in a round with $m > k$ pending players it is $\mathbb{E}[X_m] = k(1 - 1/m)^{m-1}$.
	
	Define the function $f(m)$ to be the expected finishing time when there are $m > k$ pending players. We assume that this function is non-decreasing and concave. Then we have 
	\begin{align}\label{f(m)2}
	f(m) &= 1 + \sum_{i=0}^{m}\text{Pr}(X_m=i)f(m-X_m)  \nonumber	\\
	&= 1 + \mathbb{E}[f(m-X_m)]  \nonumber  \\
	&\leq 1 + f\left(\mathbb{E}[m-X_m]\right)	\nonumber \qquad &\text{(concavity of $f$ and Jensen's inequality)} \\
	&= 1 + f(m- \mathbb{E}[X_m])	 \qquad &\text{(linearity of expectation)}
	\end{align}
	
	Now by exploiting the monotonicity of the function $f(m)$ in equation (\ref{f(m)2}), and using Lemma \ref{balls-bins} we only need to find a lower bound on $\mathbb{E}[X_m]$. This is easy, since $k \left(1 - \frac{1}{m}\right)^{m-1} \geq k/e$. Then from equation (\ref{f(m)2}) we get
	\begin{align*}
	f(m) & \leq 1 + f\left(m - \frac{k}{e}\right)	\\
	& \leq r + f\left(m - r \frac{k}{e} \right).
	\end{align*}	
	We use as base case $f(k)$ for which holds that $f(k) \leq \frac{1}{1-\ln (e-1)} \ln (\frac{k}{2}) + \left( 1 - \frac{1}{k} \right)^{-1}$, due to Lemma \ref{claim:finish-time(n,n)}. Then the $r$ for which $m - r \frac{k}{e} = k$ finally gives us 
	\begin{align*}
	f(m) \leq  e \frac{m-k}{k} + \frac{1}{1-\ln (e-1)} \ln (\frac{k}{2}) + \left( 1 - \frac{1}{k} \right)^{-1}.
	\end{align*}	
\end{proof}

Using the above lemmata we are able to prove the following.

\begin{lemma}\label{ternary-whp}
	(a) If at $t=0$ there are $n$ pending players, the probability that more than $k$ players are pending at time $t_1=2e(n-k)/k$ is at most exp$\left(-\frac{n-k}{2ek} \right)$. \\
	(b) If at $t=0$ there are $k$ pending players, the probability that not all players have transmitted successfully at time $t_2=2e(n-k)/k$ is at most exp$\left(-\frac{n-k}{2ek} \right)$.
\end{lemma}

\begin{proof}
	Let $\{Y_t\}_{t=1}^{t_1}$ be random variables which indicate the number of successful transmissions that occur in each time-step from $t=1$ up to $t_1 \triangleq 2e(n-k)/k$, given that there are $n$ pending players at time $t=0$. For the events for which $Y \triangleq \sum_{t=1}^{t_1} Y_t > n - k$ we have the desired outcome. For the rest, since the pending players in each round $1 \leq t \leq t_1$ are $m > k$, the protocol prescribes to each player probability $1/m$ on each channel. Therefore, by Lemma \ref{balls-bins}, we have $\mathbb{E}[Y_t] = k\left( 1 - 1/m \right)^{m-1}$. In the next claim we show that $Y_t$ stochastically dominates a random variable $Z_t \in \{0,1,\dots,k\}$ that indicates the number of successful transmissions in round $1 \leq t \leq t_1$ but, in this process, the players that transmit successfully are placed back to the group of pending players. 
	
	\begin{claim}\label{dominance_Y-Z}
		Pr\{$Y_t \geq x$\} $\geq$ Pr\{$Z_t \geq x$\}, for all $x \in \{0,1,\cdots, k\}$.
	\end{claim}
	
	\begin{proof}
		We will prove the above claim by showing the stronger fact that, for any fixed number $1 \leq m \leq n-1$ of pending players at time $t$, 
		\begin{align*}
		&\text{Pr}\{Y_t \geq x\text{ }|\text{ }m \text{ pending players}\} \geq \text{Pr}\{Y_t \geq x\text{ }|\text{ }m+1 \text{ pending players}\},
		\end{align*}  
		for all  $x \in \{0,1,\cdots, k\}$.
		
		Indeed, by substituting the probabilities of the above inequality we get,
		\begin{align*}
		\binom{m}{x} x! \left( \frac{1}{m} \right)^x \left( 1 - \frac{x}{m} \right)^{m-x} \geq \binom{m+1}{x} x! \left( \frac{1}{m+1} \right)^x \left( 1 - \frac{x}{m+1} \right)^{m+1-x}, \\
		\text{or equivalently,} \quad (m+1)^m (m-x)^{m-x} \geq m^m (m-x+1)^{m-x}, \\
		\text{and finally,} \quad \left(1+ \frac{1}{m}\right)^m \geq \left(1+ \frac{1}{m-x}\right)^{m-x},
		\end{align*}  
		which is true, since the function  $f(w)=\left( 1 + 1/w \right)^w $ is strictly increasing. The claim follows from the fact that for any fixed $x \in \{0,1,\cdots, k\}$,
		~\\
		~\\
		$\text{\qquad  \qquad  \qquad  \qquad  \quad Pr}\{Z_t \geq x\} = \text{Pr}\{Y_t \geq x\text{ }|\text{ }n \text{ pending players}\}$.
	\end{proof}
	
	Clearly $\{Z_t\}_{t=1}^{t_1}$ are independent random variables bounded in $[0,k]$. Let $Z \triangleq \sum_{t=1}^{t_1} Z_t$ and $\mu_1 \triangleq \mathbb{E}[Z] = \sum_{t=1}^{t_1} \mathbb{E}[Z_t] = t_{1} k\left( 1 - 1/n \right)^{n-1}$. Then by Hoeffding's inequality \cite{H63} and the stochastic domination we have,
	\begin{align*}
	\text{Pr}(Y \leq n-k) & \leq \text{Pr}(Z \leq n-k) = \text{Pr}\left( Z \leq \frac{\mu_1}{2e\left( 1 - 1/n \right)^{n-1}} \right) \leq \text{Pr}\left( Z \leq \frac{\mu_1}{2} \right) \\
	& \leq exp\left( -\frac{(1 - 1/2)^2 \mu_{1}^2}{t_{1} (k-0)^2} \right) \leq exp\left( -\frac{1}{4}\frac{t_{1}}{e^2} \right) = exp\left( -\frac{n-k}{2ek} \right),	
	\end{align*}
	where in the last three inequalities we used the fact that $(1-1/n)^{n-1} \geq 1/e$.
	
	For the second part of the proof, suppose the process is at round $t=0$ with $k$ pending players. Let $\{X_t\}_{t=1}^{t_2}$ be random variables which indicate the number of successful transmissions that occur in each time-step from $t=1$ up to $t_2 \triangleq 2e(n-k)/k$, given that there are $k$ pending players at time $t=0$. The pending players in each round $1 \leq t \leq t_2$ are $m \leq k$, hence the protocol prescribes to each player probability $1/k$ on each channel. By Lemma \ref{balls-bins}, we have $\mathbb{E}[X_t] = m\left( 1 - 1/k \right)^{m-1}$. Now, observe that $X_t$ stochastically dominates a random variable $W_t \in \{0,1,\dots,k\}$ that indicates the number of successful transmissions in round $1 \leq t \leq t_2$ but, in this process, the players that transmit successfully are placed back to the group of pending players. The latter observation is easy to see since an argument similar to the Claim that was stated earlier holds in this case. 
	
	Clearly, $\{W_t\}_{t=1}^{t_2}$ are independent random variables bounded in $[0,k]$. Let $W \triangleq \sum_{t=1}^{t_2} W_t$ and $\mu_2 \triangleq \mathbb{E}[W] = \sum_{t=1}^{t_2} \mathbb{E}[W_t] = t_{2} k\left( 1 - 1/k \right)^{k-1}$. Then by Hoeffding's inequality \cite{H63} and the stochastic domination we have,
	\begin{align*}
	\text{Pr}(X \leq k -1) &\leq \text{Pr}(W \leq k) = \text{Pr}\left( W \leq \frac{\mu_{2} k}{2e(n-k)\left( 1 - 1/k \right)^{k-1}} \right)  \\
	& \leq \text{Pr}\left( W \leq \frac{\mu_2}{2} \right) \leq exp\left( -\frac{(1 - 1/2)^2 \mu_{2}^2}{t_{2} (k-0)^2} \right) \leq exp\left( -\frac{1}{4}\frac{t_{2}}{e^2} \right) \\
	&= exp\left( -\frac{n-k}{2ek} \right),	
	\end{align*}
	where in the last three inequalities we used the fact that $(1-1/k)^{k-1} \geq 1/e$, and $n \geq 2k + 1$. This completes the proof of the lemma.
\end{proof}

We define the following anonymous protocol. In the next theorem we show that it is an equilibrium protocol and also that it is efficient.

\begin{tcolorbox}
	\textbf{\underline{Protocol $\mathbf{\emph{r}}$}}: \nonumber \\ 
	\text{Let the deadline be $t_0 = 4e(n-k)/k$. Every player among $1 \leq m \leq n$ pending players for} \\ 
	\text{$1 \leq t \leq t_0 - 1$ assigns transmission probability $1/\max\{m, k\}$ to each channel. Right before} \\ 
	\text{$t_0$ each pending player is assigned to a random channel equiprobably, and for $t \geq t_0$ always} \\ 
	\text{attempts transmission to that channel.}  
\end{tcolorbox}

\begin{theorem}\label{ter: efficient}
	Protocol $r$ for $n \geq 2k + 1$ players and $k \geq 1$ channels is an efficient, equilibrium protocol.
\end{theorem}

\begin{proof}
	First, we show that it is an equilibrium protocol when $n \geq 2k + 1$. The expected latency of a player using protocol $r$ is $\infty$. That is because there is an event with positive probability in which some player $i$ finds herself in an equilibrium where at least 2 of the other players have been assigned to each and all of the $k$ channels and transmit there in every time slot. In particular, with probability at least $k(\frac{1}{n})^{t_0 - 1} > 0$ all players will be pending right after $t_0 - 1$. Given this, with probability $\binom{n-1}{2,2,...,2,n-1-2k} (\frac{1}{k})^{n-1} > 0$ exactly 2 out of $n-1$ players will be assigned to each of the $k-1$ channels and the remaining players (including player $i$), which are at least $3$, are assigned to the remaining channel. Therefore, the aforementioned two events occur with positive probability, and then for player $i$ all channels are blocked for every $t \geq t_0$, resulting to infinite latency. Hence, the expected latency of a player using protocol $r$ is $\infty$. 
	
	Now suppose that player $i$ unilaterally deviates to some protocol $r'$. The event that all players are pending right before $t_0$ remains non-empty, since the event that all players transmit on the same channel as $i$ for every $1 \leq t \leq t_0 - 1$ happens with positive probability. Given that, the event that at least 2 of the players other than $i$ will be assigned to each channel happens with positive probability. Therefore, the deviator's expected latency remains $\infty$ and $r$ is an equilibrium protocol. 
	
	Finally, we will show that, when $n \in \omega(k)$, this protocol is also efficient, i.e. the time until all $n$ players transmit successfully is linear in $n/k$ with probability tending to 1 as $\frac{n}{k} \to \infty$. By Lemma \ref{ternary-whp}, the probability that not all players have successfully transmitted by time $t_1 + t_2 = 4e(n-k)/k$ is at most $\text{exp}\left(-\frac{n-k}{2ek} \right) + \text{exp}\left(-\frac{n-k}{2ek} \right) = 2 \text{exp}\left(-\frac{n-k}{2ek} \right).$
	Therefore, when $n \in \omega(k)$, no player is pending after $4e(n-k)/k$ rounds with high probability.
\end{proof}

\section{Open Problems}
This work leaves open some interesting problems. One of them is to find equilibria for arbitrary number of players in the multiple-channel setting with acknowledgement-based feedback. This will probably require a characterization of equilibria such as the one we provide for ternary feedback protocols in Subsection \ref{ter: history-dep char}. 

Another important open problem is to prove or disprove that there exists a FIN-EQ protocol that is efficient in the multiple-channels setting. This could be a deadline protocol or it might use some other key idea to impose a heavy latency on the players as a threat, so that they auto-restrain themselves from frequently attempting transmission. Proving that there is no efficient deadline FIN-EQ for the multiple-channel setting would be an interesting ``paradox'', since an efficient deadline FIN-EQ is found in \cite{FMN07} for the single-channel setting with ternary feedback. In view of Theorem \ref{thm: impos} we conjecture that the ``paradox'' is there.

~\\
\textbf{Acknowledgements.}
We would like to thank Frans Oliehoek for useful discussions on Partially Observable Markov Decision Processes.




\bibliographystyle{abbrv}
\bibliography{references}


\end{document}